\documentclass[runningheads]{llncs}
\usepackage{graphicx}
\usepackage{caption}

\usepackage{amsmath}
\usepackage{amsmath,amssymb,amsfonts}
\usepackage{mathtools}
\usepackage{multirow}
\usepackage{setspace}

\begin{document}

	\title{Analysis of error dependencies on NewHope	\thanks{The full citation to the IEEE published work with a Digital Object Identifier (DOI) 10.1109/ACCESS.2020.2977607}}
	
	%
	%
	\author{Minki~Song\inst{1}\and
		Seunghwan~Lee\inst{1} \and
		Eunsang~Lee\inst{2} \and
		Dong-Joon~Shin\inst{1}	\and	
		Young-Sik~Kim\inst{3} \and
		Jong-Seon~No\inst{2}
		}
	\authorrunning{}
	%
	\institute{Hanyang University, Seoul, Korea \and
	Seoul National University, Seoul, Korea \and
	Chosun University, Gwang-ju, Korea \\
	\email{minkisong@hanyang.ac.kr, kr3951@hanyang.ac.kr, eslee3209@ccl.snu.ac.kr, djshin@hanyang.ac.kr, iamyskim@chosun.ac.kr, jsno@snu.ac.kr}}
	
\maketitle              

\begin{abstract}
Among many submissions to the NIST post-quantum cryptography (PQC) project, NewHope is a promising key encapsulation mechanism (KEM) based on the Ring-Learning with errors (Ring-LWE) problem.
Since NewHope is an indistinguishability (IND)-chosen ciphertext attack secure KEM by applying the Fujisaki-Okamoto transform to an IND-chosen plaintext attack secure public key encryption, accurate calculation of decryption failure rate (DFR) is required to guarantee resilience against attacks that exploit decryption failures.
However, the current upper bound on DFR of NewHope is rather loose because the compression noise, the effect of encoding/decoding of NewHope, and the approximation effect of centered binomial distribution are not fully considered. 
Furthermore, since NewHope is a Ring-LWE based cryptosystem, there is a problem of error dependency among error coefficients, which makes accurate DFR calculation difficult.
In this paper, we derive much tighter upper bound on DFR than the current upper bound using constraint relaxation and union bound.
Especially, the above-mentioned factors are all considered in derivation of new upper bound and the centered binomial distribution is not approximated to subgaussian distribution.
In addition, since the error dependency is considered, the new upper bound is much closer to the real DFR than the  previous upper bound.
Furthermore, the new upper bound is parameterized by using Chernoff-Cramer bound in order to facilitate calculation of new upper bound for the parameters of NewHope.
Since the new upper bound is much lower than the DFR requirement of PQC, this DFR margin is used to improve the security and bandwidth efficiency of NewHope.
As a result, the security level of NewHope is improved by 7.2 \% or bandwidth efficiency is improved by 5.9 \%.
This improvement in the security and bandwidth efficiency can be easily achieved because there is little change in time/space complexity of NewHope.




	\keywords{Bandwidth~Efficiency \and Chernoff-Cramer~Bound \and Decryption~Failure~Rate \and Error~Dependency \and NewHope \and NIST \and Post-Quantum~Cryptography \and Relaxation \and Security \and Union~Bound \and Upper~Bound}
\end{abstract}

\section{Introduction}


Current public-key algorithms based on integer decomposition, discrete logarithm, and elliptic curve discrete logarithm problems (e.g, RSA and elliptic curve cryptography) have been unlikely to be broken by currently available technology.
However, with the advent of quantum computing technology such as Shor's quantum algorithm for integer factorization, current public-key algorithms can be easily broken.
For that reason, in order to avoid such security problem of future systems, new public-key algorithms called post-quantum cryptography (PQC) should be developed to replace the existing public-key algorithms.
Therefore, the National Institute of Standards and Technology (NIST) has recently begun a PQC project to identify and evaluate post-quantum public-key algorithms secure against quantum computing \cite{NEON1}.
Among the various PQC candidates, lattice-based cryptosystems have become one of the most promising candidate algorithms for post-quantum key exchange.
Lattice-based cryptosystems have been developed based on worst-case assumptions about lattice problems that are believed to be resistant to quantum computing.

Among various lattice problems, learning with errors (LWE) problem introduced by Regev in 2005 \cite{LWE} has been widely analyzed and used.
Furthermore, the Ring-LWE problem presented by Lynbashevsky, Peikert, and Regev in 2010 \cite{RLWE}, which improves the computational and implementation efficiency of LWE, has also been widely used \cite{NewHope NIST}, \cite{Frodo}, \cite{Kyber}, \cite{LAC}, \cite{HILA5}. 
NewHope has been proposed by Alkim, Ducas, P{\"o}ppelmann, and Schwabe \cite{NewHope1}, \cite{NewHope2}, which is one of the various cryptosystems based on Ring-LWE.
NewHope has attracted a lot of attention \cite{NewHope Gaborit}, \cite{NEON}, \cite{Fritzmann} and it was verified in an experiment of Google \cite{NEON 8}.
The key reasons that NewHope attracts so much attention are the use of simple and practical noise distribution, a centered binomial distribution, and a proper choice of ring parameters for better performance and security.


NewHope is an indistinguishability (IND)-chosen ciphertext attack (CCA) secure key encapsulation mechanism (KEM) that exchanges the shared secret key based on the IND-chosen plaintext attack (CPA) secure public-key encryption (PKE).
Note that the IND-CPA secure PKE can be transformed into the IND-CCA secure KEM using Fujisaki-Okamoto (FO) transform \cite{FO}.
The IND-CCA secure KEM obtained by applying FO transform to IND-CPA secure PKE requires a very low decryption failure rate (DFR) because an attacker can exploit the decryption failure \cite{FO}, \cite{FO attack}.
Therefore, the DFR of NewHope should be lower than $2^{-128}$ to make sure of resilience against attacks that exploit decryption failures.
Note that as in Frodo \cite{Frodo} and Kyber \cite{Kyber}, this study aims to achieve the DFR lower than $2^{-140}$ to allow enough margin in NewHope.
In \cite{NewHope NIST}, \cite{NewHope1}, an upper bound on DFR of NewHope is derived but this upper bound on DFR is rather loose because the compression noise, the effect of encoding/decoding of NewHope, and approximation effect of centered binomial distribution are not fully considered.
Furthermore, according to \cite{ErrorDependency}, \cite{ErrorDependency2}, accurate calculation of DFR is difficult because there is a problem of error dependency in Ring-LWE based cryptosystems.
However, the DFR of IND-CCA secure KEM obtained by applying FO transform to IND-CPA secure PKE must be calculated as accurately as possible because DFR is closely related to the security.

In this paper, an upper bound on DFR of NewHope, which is much closer to the real DFR than the previous upper bound on DFR derived in \cite{NewHope NIST}, \cite{NewHope1}, is derived by considering the above-ignored factors.
Also, the centered binomial distribution is not approximated to the subgaussian distribution.
Especially, the new upper bound on DFR considers the error dependency among error coefficients by using the constraint relaxation, which is an approximation of a difficult problem to a nearby problem that is easier to solve, and union bound.
Furthermore, the new upper bound is parameterized by using Chernoff-Cramer (CC) bound in order to facilitate calculation of new upper bound for the parameters of NewHope.
Since the new upper bound on DFR of NewHope is much lower than the DFR requirement of PQC, this DFR margin is used to improve the security and bandwidth efficiency, which is reducing the ciphertext size.
  



\subsubsection{Contributions}
The contributions of this paper is divided into three categories.

\textbf{(1) Understanding NewHope as a Digital Communication System}
NewHope can be understood as a digital communication system.
Bob and Alice are transmitter and receiver, respectively, and the 256-bit shared secret key is a message bit stream. 
The difference between the encoding output $v$ and the received signal $v''$ distorted by many factors can be modeled as a digital communication channel.
We analyze all the noise sources of this channel and numerically calculate the noise distribution of NewHope.
Also, we analyze the encoding/decoding of additive threshold encoding (ATE) in NewHope, which is an error-correcting code (ECC) for NewHope.


\textbf{(2) DFR Analysis of NewHope By Considering Error Dependency}
The previous upper bound on DFR of NewHope \cite{NewHope NIST}, \cite{NewHope1} is rather loosely derived because the compression noise, effect of encoding/decoding of ATE in NewHope, effect of error dependency among error coefficients, and approximation effect of the centered binomial distribution are not fully considered. 
However, we derive a much closer upper bound on DFR to the real DFR than the previous upper bound on DFR by considering the above factors ignored in the derivation of previous upper bound \cite{NewHope NIST}, \cite{NewHope1}.
Also, the centered binomial distribution is used for deriving the upper bound on DFR without approximating it to the subgaussian distribution.
As a result, a new upper bound on DFR is derived, which is less than $2^{-418}$ for $n=1024$ and $2^{-399}$ for $n=512$.
Note that the previous upper bound on DFR of NewHope is less than $2^{-216}$ for $n=1024$ and $2^{-213}$ for $n=512$.

\textbf{ (3) Improvement of Security and Bandwidth Efficiency of NewHope By Using New DFR Margin}
Since the new upper bound on DFR of NewHope is much lower than the required $2^{-128}$, this DFR margin can be exploited to improve the security level by 7.2 \% or bandwidth efficiency by 5.9 \% without changing the procedures of NewHope.

\section{NewHope}
\subsection{Parameters}
There are three important parameters in NewHope: $n$, $q$, and $k$.
\begin{itemize}
\item $n$: the dimension $n=512$ or $1024$ for NewHope guarantees the security properties of Ring-LWE and enables efficient number theoretic transform (NTT) \cite{NTT}. 
\item $q$: the modulus $q=12289$ is determined to support security and efficient NTT and it is closely related with the bandwidth. 
\item $k$: the noise parameter $k=8$ is the parameter of centered binomial distribution, which determines the noise strength and hence directly affects the security and DFR \cite{NewHope NIST}.
\end{itemize}


\subsection{Notations}
\begin{itemize}
\item $\mathcal{R}_q=\mathbb{Z}_q[x]/(X^n+1)$: the ring of integer polynomials modulo $X^n+1$ where each coefficient is reduced modulo $q$.
\item $a \xleftarrow{\text{\$}} \chi $: the sampling of $a \in \mathcal{R}_q$ following the probability  distribution $\chi$ over $\mathcal{R}_q$.
\item $\psi_{k}$: the centered binomial distribution with parameter $k$, which is practically realized by $\sum_{i=0}^{k-1}(b_i-b_i')$, where $b_i$ and $b_i'$ are uniformly and independently sampled from $\{0,1\}$. The variance of $\psi_{k}$ is $k/2$ \cite{NewHope NIST}.
\item $a \circ b$: the coefficient-wise product of polynomials $a$ and $b$.
\end{itemize}

\subsection{NewHope Protocol} 
NewHope is a lattice-based KEM for Alice (Server) and Bob (Client) to share 256-bit secret key with each other. 
The protocol of NewHope is briefly explained based on Fig. \ref{NewHopeProtocol} as follows, where the functions are the same ones as defined in \cite{NewHope NIST}. \\

\begin{figure}
\centering
\includegraphics[width=\textwidth]{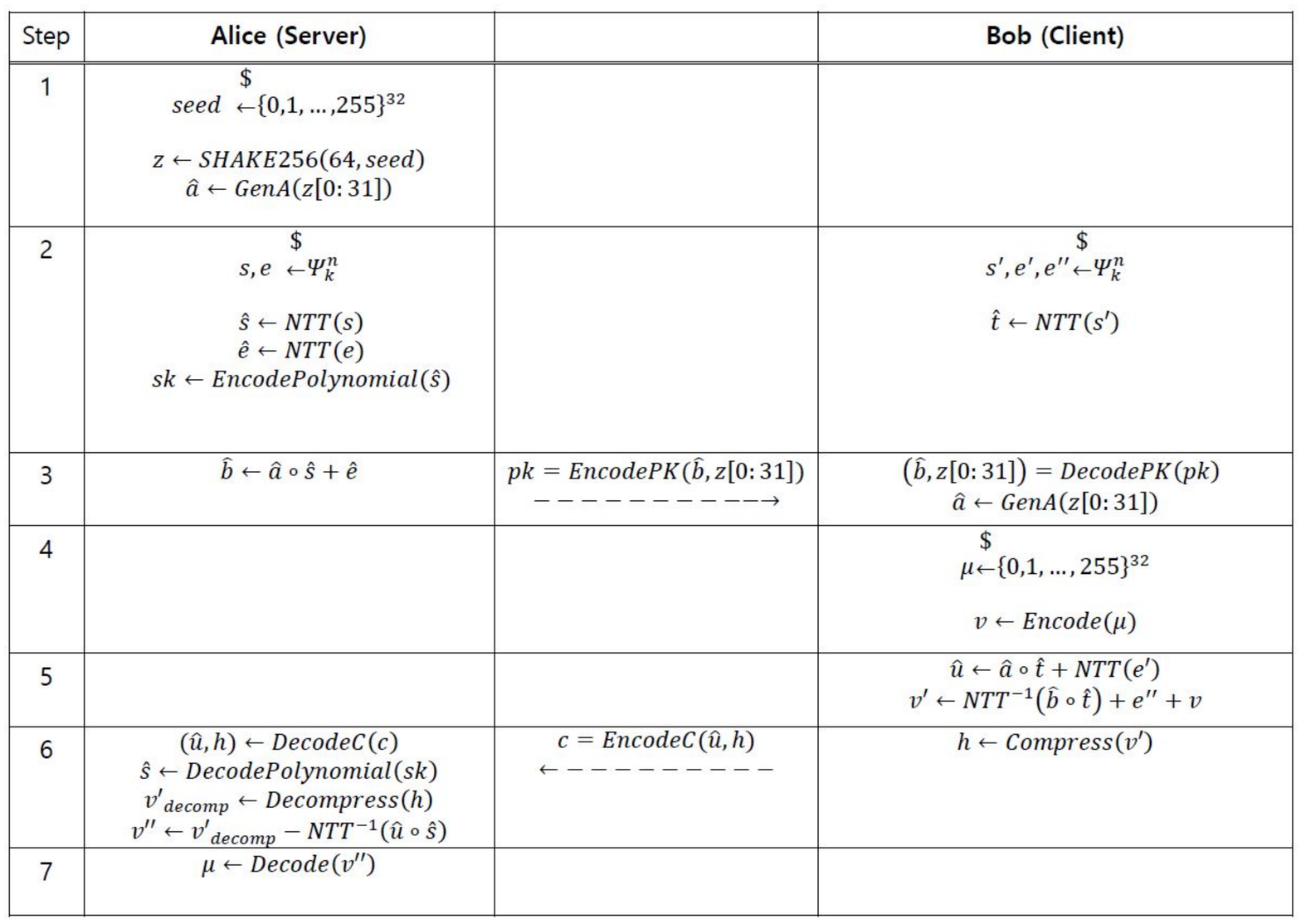}
\caption{NewHope Protocol.} 
\label{NewHopeProtocol}
\end{figure}

Step 1) $seed \xleftarrow{\text{\$}} \{0,1,\dots, 255\}^{32} $ denotes a uniform sampling of 32 byte arrays (corresponding to 256 bits) with 32 integer elements selected between 0 and 255 by using a random number generator. 
Then $SHAKE256(l,d)$, a strong hash function \cite{SHAKE256}, takes an integer $l$ that specifies the number of output bytes and a byte array $d$ as its input. 
In NewHope, $z \leftarrow SHAKE256(64,seed)$ denotes that 32 byte arrays ($seed$) are hashed to generate 64 pseudorandom byte arrays ($z$) with 64 integer elements uniformly selected between 0 and 255.
Then $GenA$ expands 32 pseudorandom byte arrays $z[0:31]$ using $SHAKE128$ hash function \cite{SHAKE256} to generate the polynomial $\hat{a} \in \mathcal{R}_q $ where $z[0:31]$ is the first 32 byte arrays of $z$.
Since $\hat{a}$ is generated from the $seed$ sampled following a uniform distribution, the coefficients of $\hat{a}$ also follow a uniform distribution on $[0,q-1]$.

Step 2) Generate polynomials ($s$, $s'$, $e$, $e'$, $e''$ $\in \mathcal{R}_q$) whose coefficients are sampled following the centered binomial distribution $\psi_k$. 
The polynomials ($s$, $s'$, $e$) are transformed to ($\hat{s}$, $\hat{t}$, $\hat{e}$), respectively, by applying NTT for efficient polynomial multiplication.
Then Alice transforms the secret key ($ \hat{s} $) into byte arrays using $ EncodePolynomial () $ which converts the polynomial ($ \hat{s} $) into 2048 byte arrays.

Step 3) Alice creates a public key ($pk$) by converting $ \hat{b}  = \hat{a} \circ \hat{s} + \hat{e}$ and $z[0:31]$ into 1824 byte arrays by using $ EncodePK () $, and transmits ($pk$) to Bob.
Then Bob transforms the received public key ($pk$) into ($ \hat{b} $, $z[0:31]$) using $ DecodePK () $, and creates ($ \hat{a} $) which is the same ($\hat{a}$) generated in Step 1.

Step 4) A 256-bit shared secret key ($\mu$) is created  and encoded by ATE encoder to generate a 1024-symbol codewords $v$. 

Step 5) Generate a ciphertext ($\hat{u}$, $v'$) by using the public key components $\hat{b}$, $\hat{a}$ ,the various errors $\hat{t}$, $e'$, $e''$  and $v$.

Step 6) To efficiently reduce bandwidth, compression is performed on the coefficients of $ v '$ to generate the polynomial $h$, and then the ciphertext polynomials ($\hat{u}$, $h$) are transformed into the byte arrays $c$ by using $EncodeC()$, and $c$ is transmitted to Alice.
Alice performs decompression on $\hat{h}$ to restore $v'$.
However, this decompressed polynomial $v'_{decomp}$ is different from $v'$ generated in Step 5, due to the loss from compression and decompression.
Alice creates $ v '' $ by using the received ciphertext $c$ and $ sk $ generated in Step 2. 
Each coefficient of $ v '' $ is a sum of the corresponding coefficients of $ v $ and errors. 
Note that $v'' $ is not a polynomial used in NewHope, but it is added in Fig. \ref{NewHopeProtocol} for easy explanation of the results in this paper.

Step 7) The 256-bit shared secret key ($ \mu $) is recovered (or decrypted) from the coefficients of $v''$ by performing the decoding of ATE.

\section{Understanding NewHope as a Digital Communication System}

\subsection{NewHope as a Digital Communication System}
In order to facilitate analysis of DFR of NewHope, it is much more convenient to understand the protocol of NewHope as a digital communication system.
For NewHope, the mapping $\mathbb{Z}^{256}_2 \rightarrow \mathbb{Z}^{n}_2$ ($\mu \rightarrow \mu_{enc}$) and the mapping $\mathbb{Z}^{n} \rightarrow \mathbb{Z}^{256}_2$ ($\mu'_{enc} \rightarrow \mu'$) through ATE, $n=512$ or $1024$, can be regarded as encoding and decoding of ECC, respectively. 
Also, the mapping $\mathbb{Z}^{n}_2 \rightarrow \mathcal{R}_q$ ($\mu \rightarrow v$) and $\mathcal{R}_q \rightarrow \mathbb{Z}^{n}$ ($v'' \rightarrow \mu'_{enc}$) through ATE can be regarded as modulation and demodulation, respectively.
Then NewHope can be understood as a digital communication system as follows.

Bob and Alice are transmitter and receiver, respectively, and the 256-bit shared secret key ($\mu$) is a message bit stream.
Also, the process of transmitting and receiving messages (Steps 4, 5, 6, and 7) can be viewed as a digital communication channel.
In more detail, the transmitter (Bob) generates a 256-bit message bit stream, encodes this massage into an $n$-bit codeword, modulates each codeword bit to a symbol of $\mathbb{Z}_q$, and transmits the resulting signal (Step 4).
At the receiver (Alice), the received signal through the noisy channel is demodulated and decoded (Step 7).
For NewHope, a process of adding the compression noise and the difference noise generated in Steps 5 and 6 can be regarded as noisy communication channel.
This overall process in Steps 4-7 can be described as a digital communication system shown in Fig. \ref{Protocol0}.
\begin{figure}[t]
\centering
\includegraphics[width=\textwidth]{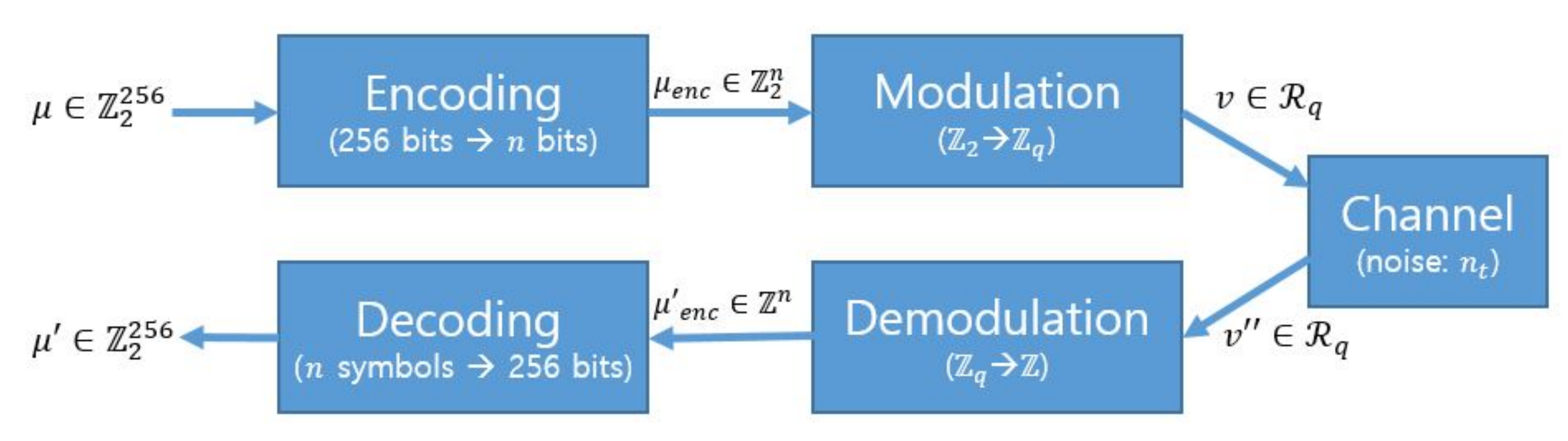}
\caption{An interpretation of NewHope as a digital communication system ($n=512$ or $1024$).} 
\label{Protocol0}
\end{figure}

In Fig. \ref{Protocol0}, $\mu_{enc}$ is the encoded signal of $\mu$ by applying encoding of ATE, and $n_t$ represents the overall noise generated in Steps 5 and 6, which is called the total noise $n_t$.
After interpreting NewHope as a digital communication system, the DFR in NewHope is equivalent to the block error rate $Pr(\mu \neq \mu')$ in a digital communication system.
Therefore, in order to calculate tight upper bound on DFR of NewHope, exact analysis of encoding/modulation and decoding/demodulation of NewHope and the noisy channel is required.
In the following sebsection 3.2, each operation in Fig. \ref{Protocol0} is explained in detail and analyzed.

\subsection{Analysis of Encoding/Modulation and Decoding/Demodulation and Channel Noise of NewHope}

\subsubsection{Analysis of Encoding/Modulation and Decoding/Demodulation of NewHope: ATE}
In NewHope, ATE is used to encode and modulate a message bit $\mu_i$, and decode and demodulate an erroneous message bit $ v''_i $.
Note that ATE performs both encoding/decoding as an ECC and modulation/demodulation.
The encoding/modulation and decoding/demodulation procedures of ATE with $m$ repetitions are shown in Fig. \ref{ATE encoding decoding} where $m=4$ for $n=1024$ and $m=2$ for $n=512$ \cite{ATE}.
The encoding of ATE is performed such that one message bit $\mu_i$ is repeated $m$ times and the modulation of ATE is a mapping of each bit to an element of $\mathbb{Z}_q$ (usually either 0 or $\lfloor \frac{q}{2} \rfloor$) as the coefficients of $v$.
Note that the $m$-repetition is the same operation as the encoding of an $m$-repetition code.
The demodulation of ATE is to calculate the absolute value of the difference between the received erroneous symbol $v''_i$ and $\lfloor q/2 \rfloor$ over integer domain $\mathbb{Z}$.
The decoding of ATE is to sum up $m$ absolute values corresponding to the same $\mu'_{enc,i+256l}$, $\forall l \in [0,m-1]$ to generate $\mu'_{s,i}$ and compare it with the decision threshold $m \cdot q / 4$ to determine if the estimate $\mu'_i$ of $\mu_i$ is $0$ or $1$ as follows.
\begin{equation}
\mu'_{s,i} \underset{\mu'_i=1}{ \overset{\mu'_i=0} \gtrless} \frac{m \cdot q}{4}
\end{equation}
\begin{figure}[t]
\centering
\includegraphics[width=\textwidth]{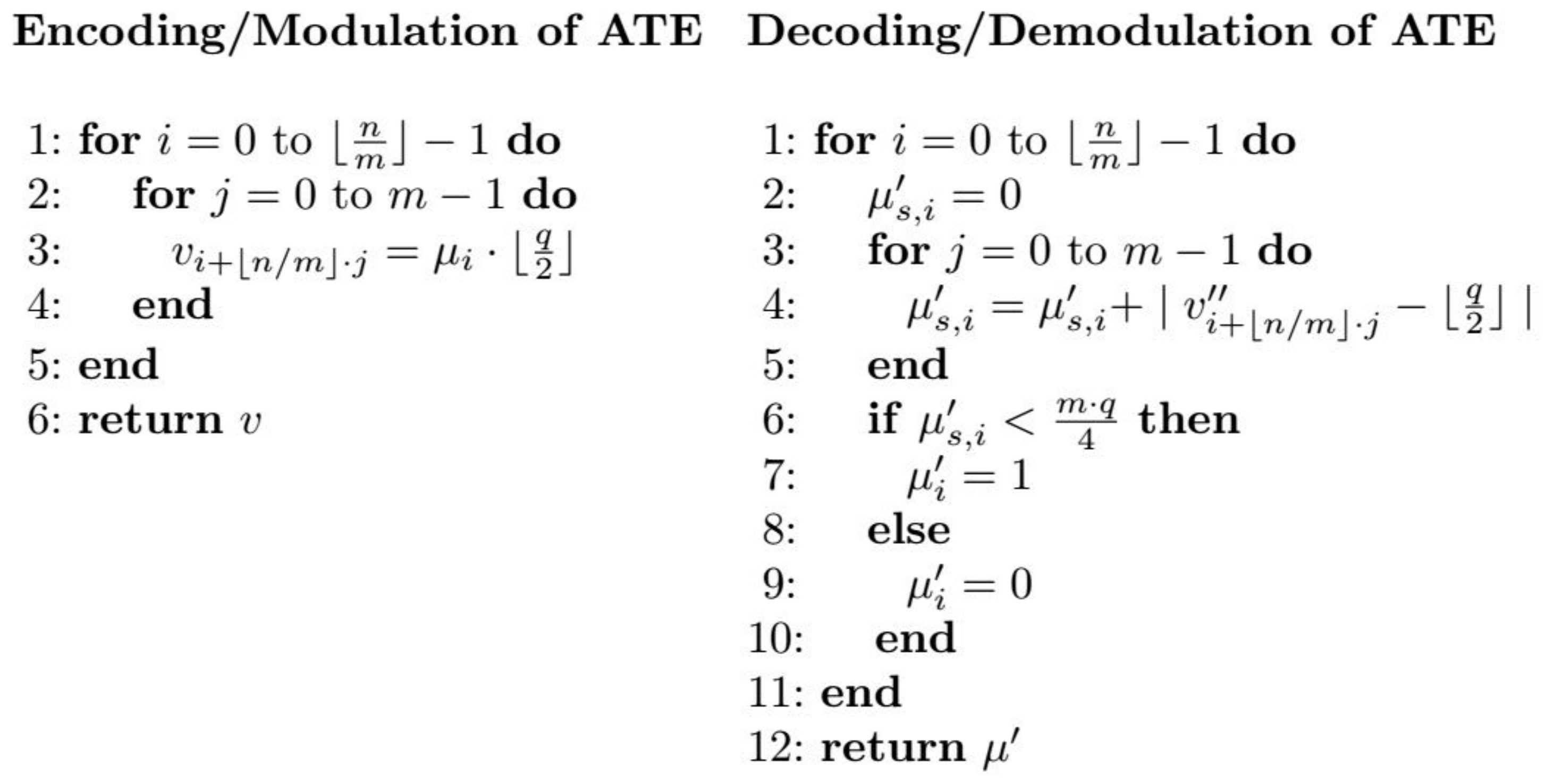}
\caption{Encoding/Modulation and Decoding/Demodulation of ATE in NewHope.} 
\label{ATE encoding decoding}
\end{figure}
\subsubsection{Analysis of Difference Noise, Compression Noise, and Total Noise of NewHope}
Total noise $n_t$ is defined as the noise contained in the received signal $v''$ except the transmitted signal $v$. 
The $i$th coefficient $n_{t,i}$ of the total noise polynomial $n_t$ contained in the polynomial $v''$ in Step 6 is expressed as follows.

\begin{eqnarray}
n_{t,i} & = & (v''-v)_i \nonumber \\
& =& (v'_{decomp}-us-v)_i  \nonumber \\
& =& (v'+n_c-us-v)_i \nonumber \\
& =& (bs'+e''-ass'-e's)_i+n_{c,i} \nonumber \\
& =& (es'-e's+e'')_i+n_{c,i} \nonumber \\
& =& n_{d,i}+n_{c,i},
\label{total noise equation}
\end{eqnarray}
where $(\cdot)_i$ denotes the $i$th coefficient of the given polynomial, $n_c \in \mathcal{R}_q$ is the compression noise polynomial, $n_{c,i}$ is the $i$th coefficient of $n_c$ contained in $v''$, $n_d \in \mathcal{R}_q$ is the difference noise polynomial, and $n_{d,i}$ is the $i$th coefficient of $n_d$ contained in $v''$.

To analyze the compression noise $n_{c,i}$, we first need to investigate the coefficient of the polynomial $v'=ass '+ es' + e ''$ being compressed, where the coefficients of $s$, $s'$, $e$, and $e''$ follow the predetermined centered binomial distribution. 
However, since the coefficients of polynomial $a$ follow a uniform distribution, the coefficient of the compressed polynomial $h$ will eventually follow a uniform distribution.
A compression to $v'$ is performed by applying $\lfloor v'_i * r/q \rceil$ to the coefficients $v'_i$ of $v'$ to generate the coefficient $h_i$ of $h$, where $\lfloor \cdot \rceil$ is a rounding function that rounds to the closest integer, $r$ denotes the compression rate on $v'$, and $r=8$ for NewHope. 
Then the range of the compressed coefficients $h_i$ of $h$ is changed from $[0, q-1]$ to $[0, r - 1]$ so that the number of bits required to store a coefficient is reduced from 14 bits ($=\lceil \log_2q \rceil$) for $v'$ to 3 bits ($=\lceil \log_2r \rceil$) for NewHope with $r=8$.
Note that the smaller the value of $r$ is, the more compression is performed.
A decompression is performed by applying $\lfloor h_i * q/r \rceil$ to each of the coefficients of $h$.
Then the coefficient takes the value from $ 0 $, $ \lfloor q / r \rceil $,$ \lfloor 2q / r \rceil $ \dots, and $ \lfloor (r-1)\cdot q / r \rceil $.
This compression and decompression are illustrated in Fig. \ref{Compression}, where the coefficients $v'_i$ of $ v '$ from different patterns (or ranges) are mapped to different $v_{decomp,i}$ values through compression and decompression.
In the end, compression and decompression can be seen as a rounding operation.
Therefore, the compression noise is inevitably generated with the maximum magnitude $\lfloor q/2r \rfloor$ and the distribution $\Pr_{n_c}(x)$ of the compression noise is derived as follows:
\begin{eqnarray}
{\textnormal Pr}_{n_c}(x) = \left\{\begin{array}{ll}
q/r, & 0 \le x \le \lceil \frac{q}{2r} \rceil-1 \\
0, & \textrm{otherwise} \\
q/r, & q-2-\lceil \frac{q}{2r} \rceil \le x \le q-1.
\end{array} \right.
\end{eqnarray}

\begin{figure}
\centering
\includegraphics[width=\textwidth]{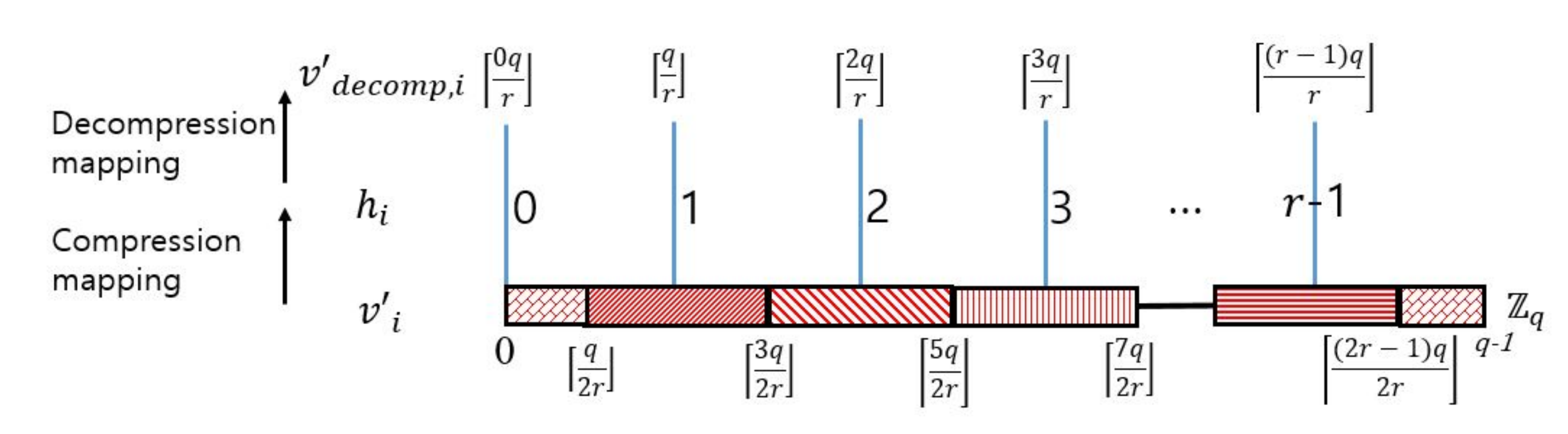}
\caption{Compression and decompression mapping in NewHope.} 
\label{Compression}
\end{figure}


To analyze the difference noise $n_{d,i}=(es'-e's + e'')_i$, we use the fact that the coefficients of $e$, $e'$, $e''$, $s$, and $s'$ are independent and identically distributed (\textit{i.i.d.}) following the same centered binomial distribution.
In order to derive the distribution of coefficient $n_{d,i}$ of $n_d$, a number of convolution operations are required because it is a sum of many \textit{i.i.d.} random variables, each of which is obtained by multiplying two \textit{i.i.d.} random variables following the centered binomial distribution.
However, since it is difficult to calculate the multiple convolutions of the above distribution in closed form, the distribution of difference noise is numerically calculated \cite{Fritzmann}.

Total noise is a sum of compression noise and difference noise which are independently generated.
Thus, the distribution of total noise is obtained by performing convolution of the distributions of compression noise and difference noise as shown in Fig. \ref{Total noise}.
However, due to the error dependency among total noise coefficients $n_{t.i}$, the distribution of only one total noise coefficient cannot be used to calculate the accurate DFR or derive a better upper bound on DFR \cite{ErrorDependency}, \cite{ErrorDependency2}.

\begin{figure}
\centering
\includegraphics[width=\textwidth]{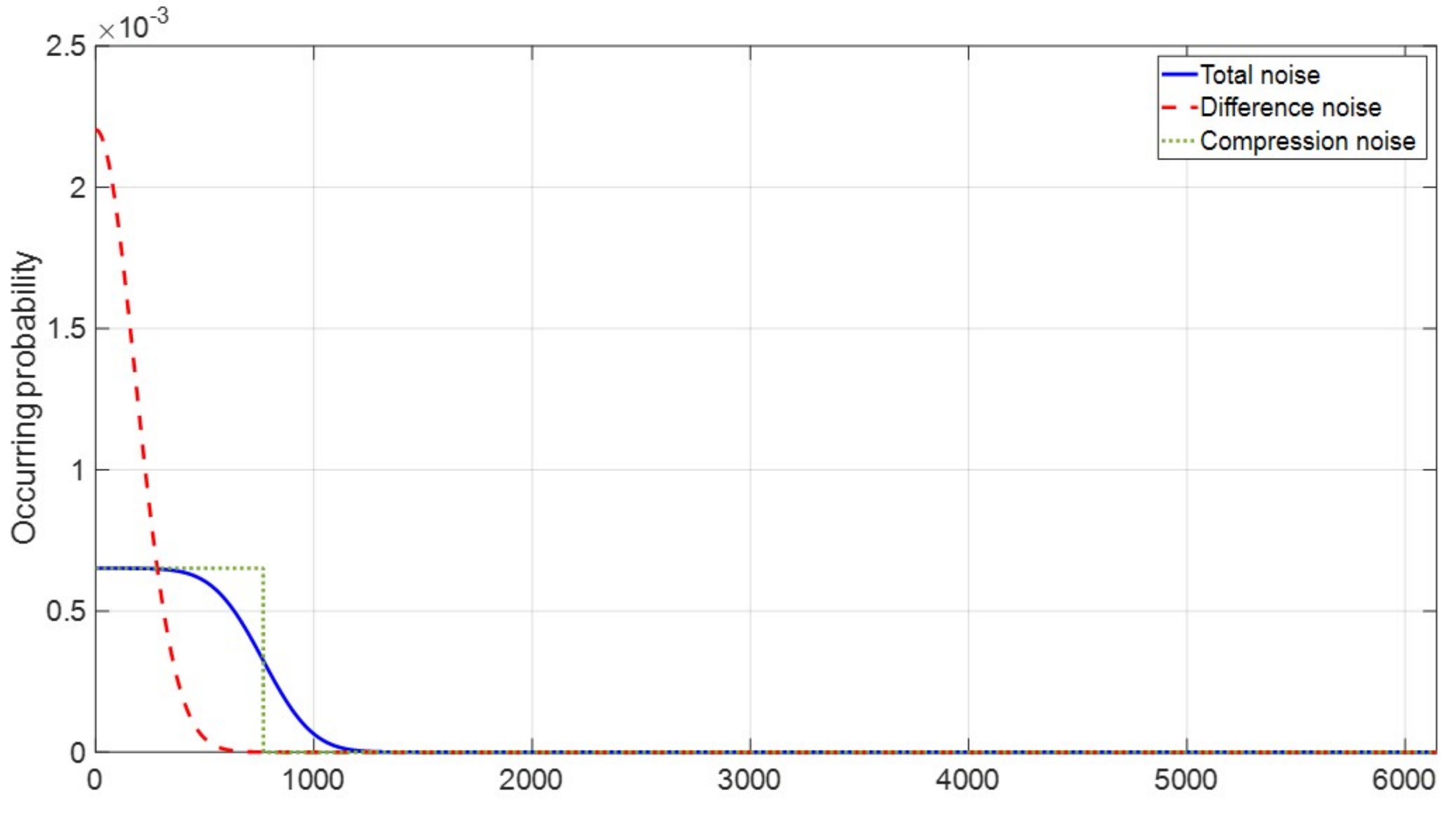}
\caption{Distributions of total noise, compression noise, and difference noise of NewHope (These distributions are symmetric with respect to $\lfloor \frac{q}{2} \rfloor$ where $q=12289$). } 
\label{Total noise}

\end{figure}

\section{DFR Analysis of NewHope By Considering Error Dependency}

In this paper, a new upper bound on DFR of NewHope, which is much tighter than the upper bound given in \cite{NewHope NIST}, \cite{NewHope1}, is derived by considering the total noise in section 3 and the centered binomial distribution without doing subgaussian approximation.
More importantly, the error dependency is considered in deriving an upper bound on DFR by using the constraint relaxation, which is an approximation of a difficult problem to a nearby problem that is easier to solve, and union bound.

A new upper bound on DFR of NewHope is derived by considering two types of error dependency as shown in Fig. \ref{dependency diagram}.
The first type of error dependency is analyzed for the output bit of one ATE decoder to derive an upper bound on the BER $\Pr(\mu_i \neq \mu_i')$.
In this case, the error dependencies among $m$ input are considered.
Note that analysis of one ATE decoder is good enough because all 256 ATE decoders are statistically identical.
The analysis of second type of error dependency is performed on 256 output bits $\mu'_i$ of ATE decoders to derive an upper bound on DFR $Pr(\mu \neq \mu')$ of NewHope.
In this case, the error dependencies among 256 bits $\mu'_i$ are considered.
\begin{figure}[t]
\centering
\includegraphics[width=\textwidth]{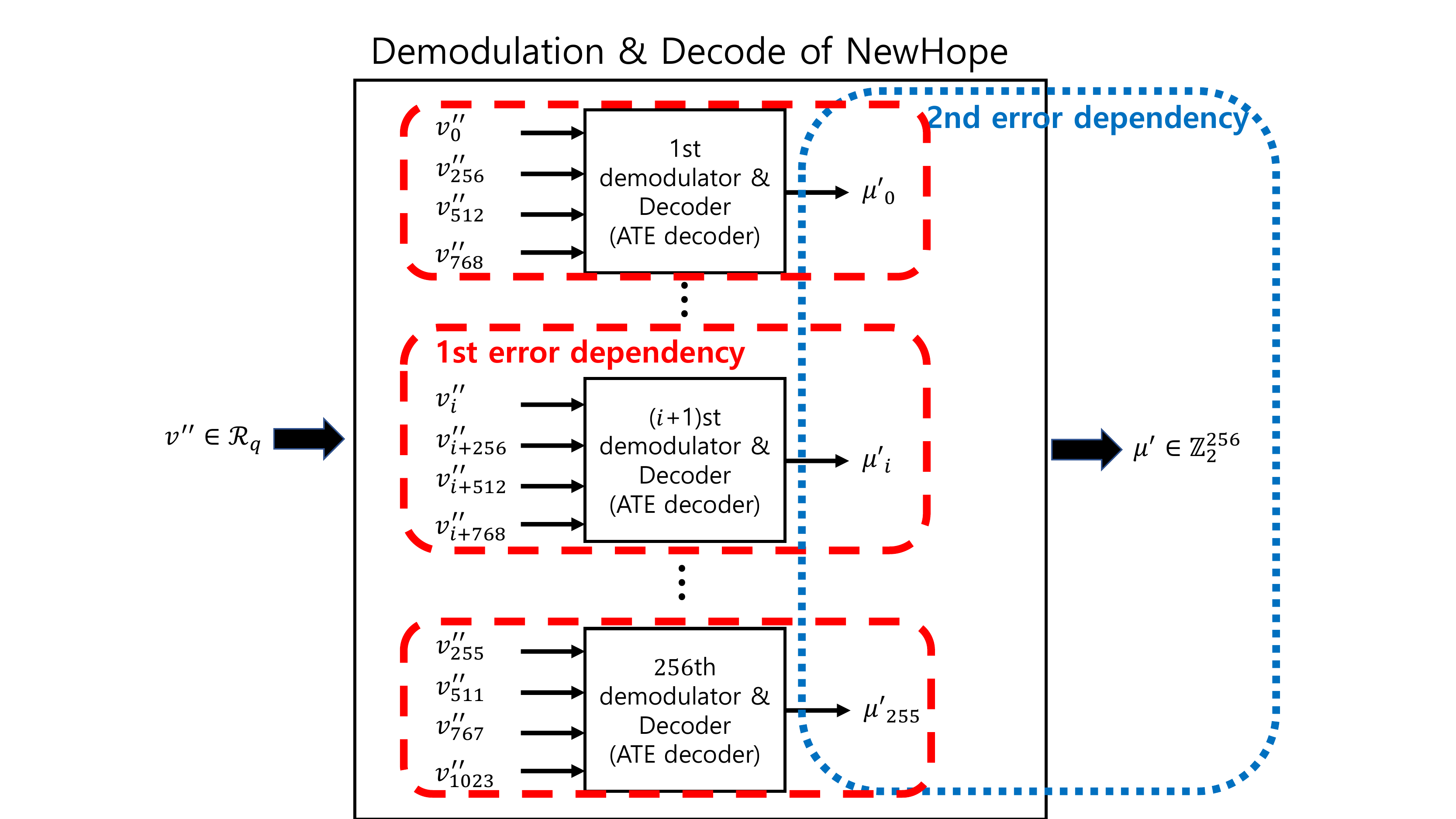}
\caption{Two types of error dependency in the demodulation and decode of NewHope.} 
\label{dependency diagram}
\end{figure}
\subsection{Propose Upper Bound on BER of NewHope}
Suppose that $\Pr(\mu_i=0) = \Pr(\mu_i=1) = 1/2$, then the BER is average of two conditional probability depending on $\mu_i$.
\begin{eqnarray}
	\Pr(\mu_i \neq \mu_i') 
	&=&
	\Pr\Big(\{\mu_i \neq \mu_i' \}\cap \{\mu_i = 0\}) +  \Pr(\{\mu_i \neq \mu_i'\} \cap \{\mu_i = 1\} \Big) \nonumber \\
	&=&{1 \over 2 }\Big(\Pr(\mu_i \neq \mu_i' | \mu_i = 0) 
	+  \Pr(\mu_i \neq \mu_i' |\mu_i = 1)\Big) 
\end{eqnarray}
Since $\Pr(\mu_i \neq \mu_i' |  \mu_i = 0)$ and $ \Pr(\mu_i \neq \mu_i' |\mu_i = 1) \nonumber $ are statistically identical, we will analysis the BER given $\mu_i = 1$. 
Then the total noise given $\mu_i=1$ is defined by $n_{t,i}^{\mu_i=1}=(n_{t,i} + \mu_{enc,i} \lfloor \frac{q}{2} \rfloor)\bmod q $ where $\mu_{enc,i}=1$.
The output $\mu'_{s,i}$ of decoding/demodulation of NewHope, which is defined in section 3.2, is determined by $m$ dependent coefficients of $v''$ given $\mu_i = 1$ as follows:
\begin{eqnarray}
	\mu_{s,i}' 
	&=& \sum_{l=0}^{m-1}
	|n_{t,i+256l}^{\mu_i=1}- \lfloor{q \over 2}\rfloor| \nonumber,\\ 
	&=& \sum_{l=0}^{m-1}
	| (n_{t,i+256l} + \lfloor \frac{q}{2} \rfloor)\bmod q - \lfloor{q \over 2}\rfloor|  
	\label{mu q}
\end{eqnarray}
where $\mu_{s,i}' \in \mathbb{Z}$.

In NewHope, most operations are performed over $\mathcal{R}_q=\mathbb{Z}_q[x]/(X^n+1)$, but for the convenience of analysis, we consider the two domains $\mathbb{Z}$ and $\mathbb{Z}_q$, and express the polynomials $e$, $s$, $e'$, $s'$, $e''$, and $n_c$ in $\mathcal{R}_q=\mathbb{Z}_q[x]/(X^n+1)$ by the vectors $\textbf{e}$, $\textbf{s}$, $\textbf{e}'$, $\textbf{s}'$, $\textbf{e}''$, and $\textbf{n}_{c}$ in $\mathbb{Z}^{n \times 1}$. 
Then, it is clear that $\textbf{e},\textbf{s},\textbf{e}',\textbf{s}',\textbf{e}'' \in \mathbb{Z}^{n \times 1}$ are the random vectors following the centered binomial distribution with the parameter $k=8$ and $\textbf{n}_{c} \in \mathbb{Z}^{n \times 1}$ is the random vector following the uniform distribution over the support $[-\lfloor{q \over 2r}\rfloor , \lfloor{q \over 2r}\rfloor]$.
To express the product of two polynomials over $\mathcal{R}_q=\mathbb{Z}_q[x]/(X^n+1)$ as an operation $\circ$ for the corresponding vectors over $\mathbb{Z}^{n \times 1}$, we define a new operation $\odot$, which is called cyclic shift product, as follows:
	
\begin{eqnarray}
(e \circ s)_i &=& (\textbf{e} \odot \textbf{s})_{i} \nonumber \\ 
&=&\sum_{j=0}^{n-1} \textnormal{sign}(i-j)e_js_{(i-j)\bmod n},
\end{eqnarray}
where $\textnormal{sign}(x)=1$ when $x\geq 0$, otherwise $\textnormal{sign}(x)=-1$. 
For examples, if $n=4$,
	\begin{eqnarray*}
	(\textbf{e} \odot \textbf{s})_{0}=
	\begin{pmatrix}
  	e_{0}\\
  	e_{1}\\
  	e_{2}\\
  	e_{3}\\
	\end{pmatrix}^T
	\begin{pmatrix}
  	+s_{0}\\
  	-s_{3}\\
  	-s_{2}\\
  	-s_{1}\\
	\end{pmatrix},
	(\textbf{e} \odot \textbf{s})_{1}=
	\begin{pmatrix}
  	e_{0}\\
  	e_{1}\\
  	e_{2}\\
  	e_{3}\\
	\end{pmatrix}^T
	\begin{pmatrix}
  	+s_{1}\\
  	+s_{0}\\
  	-s_{3}\\
  	-s_{2}\\
	\end{pmatrix}, 
	\\
	(\textbf{e} \odot \textbf{s})_{2}=
	\begin{pmatrix}
  	e_{0}\\
  	e_{1}\\
  	e_{2}\\
  	e_{3}\\
	\end{pmatrix}^T
	\begin{pmatrix}
  	+s_{2}\\
  	+s_{1}\\
  	+s_{0}\\
  	-s_{3}\\
	\end{pmatrix},
	(\textbf{e} \odot \textbf{s})_{3}=
	\begin{pmatrix}
  	e_{0}\\
  	e_{1}\\
  	e_{2}\\
  	e_{3}\\
	\end{pmatrix}^T
	\begin{pmatrix}
  	+s_{3}\\
  	+s_{2}\\
  	+s_{1}\\
  	+s_{0}\\
	\end{pmatrix},
	\end{eqnarray*}
where $(\cdot)^T$ denotes the transpose of vector.
Using the newly defined vectors $\textbf{e}$, $\textbf{s}$, $\textbf{e}'$, $\textbf{s}'$, $\textbf{e}''$, $\textbf{n}_{c}$ and operation $\odot$, $\mu_{s,i}'$ in (\ref{mu q}) can be expressed as:
\begin{eqnarray}
	\mu_{s,i}' &=& \sum_{l=0}^{m-1}	| (n_{t,i+256l} + \lfloor \frac{q}{2} \rfloor)\bmod q - \lfloor{q \over 2}\rfloor|  \nonumber \\
	&=& \sum_{l=0}^{m-1}|n_{t,i+256l}^{*} - q\alpha_{i+256l}| \label{absequation},
\end{eqnarray}
where $n_{t,i}^{*}=(\textbf{e} \odot \textbf{s}')_{i}-(\textbf{e}' \odot \textbf{s})_{i}+\textbf{e}''_{i}+\textbf{n}_{c,i}$, and $\alpha_i$ is an integer making $n_{t,i}$ be in $[-\lfloor \frac{q}{2} \rfloor, \lfloor \frac{q}{2} \rfloor]$ such that $|\alpha_i| \leq \lfloor (2nk^2+k+(q-1)/r)/q \rfloor$.
For example, if $|n_{t,i}^{*}|\leq \lfloor{q \over 2}\rfloor$, then $\alpha_{i} = 0$.
Finally, under the assumption that an all-one message bit is transmitted, the event of bit error is equivalent to the following inequality. 
\begin{eqnarray}	
T_m \leq \sum_{l=0}^{m-1} |n_{t,i+256l}^{*} - q\alpha_{i+256l}| \leq 2T_m,
\label{bit error event}
\end{eqnarray}
where $T_m= \frac{m}{2} \lfloor \frac{q}{2} \rfloor$ is the decision threshold of ATE and $2T_m$ is a maximum value of $\sum_{l=0}^{m-1}|n_{t,i+256l}^{*} - q\alpha_{i+256l}|$.

In order to find the support satisfying (\ref{bit error event}), some sets and vector should be defined.
Let $\Omega$ be the support of $\textbf{e}$, $\textbf{s}$, $\textbf{e}'$, $\textbf{s}'$, $\textbf{e}'',\textbf{n}_{c}$ where $\Omega= \textnormal{sup}(\textbf{e}$, $\textbf{s}$, $\textbf{e}'$, $\textbf{s}'$, $\textbf{e}'',\textbf{n}_{c})= \{e,s,e',s',e'',n_c| e,s,e',s',e'' \in [-k,k]^{n\times 1}, n_c \in [-\lfloor{q \over 2r}\rfloor, \lfloor{q \over 2r}\rfloor]^{n \times 1}\}$, $\textnormal{sup}(\cdot)$ denotes the support of vector, $k$ is parameter of the centered binomial distribution, and $r$ is the compression rate. such that $\Pr(\Omega)=1$. 
Let $E$ be the support of bit error where $E=\{\epsilon \in \Omega | T_m \leq \sum_{l=0}^{m-1} |n_{t,i+256l}^{*} - q\alpha_{i+256l}| \leq 2T_m\}$ such that $\Pr(\mu_i \neq \mu'_i)= \Pr(E)$. 

Since (\ref{bit error event}) is the sum of $m$ absolute values, it can be divided into $m^2$ cases by using matrix $\textbf{y}^m$.
The $\textbf{y}^m$ is a matrix that replaces $[0, 1, 2, \cdots, m^2-1]^T$ with a binary matrix 
having $m$ columns and maps each element of such matrix from $0$ to $1$ and $1$ to $-1$.
For example, for $m=4$, $\textbf{y}^4_0=(1,1,1,1)$, $\textbf{y}^4_7=(1,-1,-1,-1)$, and $y^4_{0,1}=1$ where $\textbf{y}_k^m$ and $y^m_{k,l}$ denote the $k$-row vector and $(k,l)$ element of $\textbf{y}^m$, respectively.
Then, the set $\Omega_k$ that satisfies each of $m^2$ cases of (\ref{bit error event}) can be defined as follows:
\begin{eqnarray}   
  \Omega_k=\{\omega_k \in \Omega | (n_{t,i+256l}^*-q)y_{k,l}^{m}\geq 0, \; \forall l \in [0, m-1]\},
\end{eqnarray}
where the details of $\Omega_k$ is shown in Table \ref{Sets of event of bit error}.
The $\Omega_0$, $\Omega_1$, $\cdots$, and $\Omega_{m^2-1}$ are clearly disjoint set such that $\Omega = \cup_{i=0}^{m^2-1} \Omega_{i}$ and $\Omega_{i} \cap \Omega_{j}= \emptyset$ if $i \neq j$.
If $\omega_k \in \Omega_k$, then absolute values in (\ref{bit error event}) can be replaced with $\textbf{y}^{m}_{k}$ as follows:
\begin{eqnarray}
	\sum_{l=0}^{m-1}
	|n_{t,i+256l}^{*} - q\alpha_{i+256l}|=
	\sum_{l=0}^{m-1}
	(n_{t,i+256l}^{*} - q\alpha_{i+256l})y_{k,l}^{m}.	
\end{eqnarray}
\begin{table}[t]
\centering
\caption{The details of support $\Omega_k$ for $m=4$}
{\scriptsize
\begin{tabular}{|c|c|}
\hline
Set   & Set condition                                                                                                                                                                              \\ \hline
$\Omega_0$  & $\lbrace\omega_0\in\Omega | n_{t,i}^{*} - \alpha_{i}q\geq 0,n_{t,i+256}^{*} - \alpha_{i+256}q\geq 0	,n_{t,i+512}^{*} - \alpha_{i+512}q\geq 0	,n_{t,i+768}^{*} - \alpha_{i+768}q\geq 0	\rbrace $ \\ \hline
$\Omega_1$  & $\lbrace\omega_1\in\Omega | n_{t,i}^{*} - \alpha_{i}q\geq 0,n_{t,i+256}^{*} - \alpha_{i+256}q\geq 0	,n_{t,i+512}^{*} - \alpha_{i+512}q\geq 0	,n_{t,i+768}^{*} - \alpha_{i+768}q < 0	\rbrace $   \\ \hline
$\Omega_2$  & $\lbrace\omega_2\in\Omega | n_{t,i}^{*} - \alpha_{i}q\geq 0,n_{t,i+256}^{*} - \alpha_{i+256}q\geq 0	,n_{t,i+512}^{*} - \alpha_{i+512}q< 0	,n_{t,i+768}^{*} - \alpha_{i+768}q\geq 0	\rbrace $    \\ \hline
$\Omega_3$  & $\lbrace\omega_3\in\Omega | n_{t,i}^{*} - \alpha_{i}q\geq 0,n_{t,i+256}^{*} - \alpha_{i+256}q \geq 0	,n_{t,i+512}^{*} - \alpha_{i+512}q< 0	,n_{t,i+768}^{*} - \alpha_{i+768}q< 0	\rbrace $      \\ \hline
$\Omega_4$  & $\lbrace\omega_4\in\Omega | n_{t,i}^{*} - \alpha_{i}q \geq 0,n_{t,i+256}^{*} - \alpha_{i+256}q< 0	,n_{t,i+512}^{*} - \alpha_{i+512}q\geq 0	,n_{t,i+768}^{*} - \alpha_{i+768}q \geq 0	\rbrace $      \\ \hline
$\Omega_5$  & $\lbrace\omega_5\in\Omega | n_{t,i}^{*} - \alpha_{i}q \geq 0,n_{t,i+256}^{*} - \alpha_{i+256}q< 0	,n_{t,i+512}^{*} - \alpha_{i+512}q\geq 0	,n_{t,i+768}^{*} - \alpha_{i+768}q < 0	\rbrace $     \\ \hline
$\Omega_6$  & $\lbrace\omega_6\in\Omega | n_{t,i}^{*} - \alpha_{i}q\geq 0,n_{t,i+256}^{*} - \alpha_{i+256}q < 0	,n_{t,i+512}^{*} - \alpha_{i+512}q < 0	,n_{t,i+768}^{*} - \alpha_{i+768}q\geq 0	\rbrace $     \\ \hline
$\Omega_7$  & $\lbrace\omega_7\in\Omega | n_{t,i}^{*} - \alpha_{i}q \geq 0,n_{t,i+256}^{*} - \alpha_{i+256}q < 0	,n_{t,i+512}^{*} - \alpha_{i+512}q< 0	,n_{t,i+768}^{*} - \alpha_{i+768}q < 0	\rbrace $       \\ \hline
$\Omega_8$  & $\lbrace\omega_8\in\Omega | n_{t,i}^{*} - \alpha_{i}q< 0,n_{t,i+256}^{*} - \alpha_{i+256}q \geq 0	,n_{t,i+512}^{*} - \alpha_{i+512}q\geq 0	,n_{t,i+768}^{*} - \alpha_{i+768}q\geq 0	\rbrace $       \\ \hline
$\Omega_9$  & $\lbrace\omega_9\in\Omega | n_{t,i}^{*} - \alpha_{i}q < 0,n_{t,i+256}^{*} - \alpha_{i+256}q \geq 0	,n_{t,i+512}^{*} - \alpha_{i+512}q \geq 0	,n_{t,i+768}^{*} - \alpha_{i+768}q< 0	\rbrace $          \\ \hline
$\Omega_{10}$ & $\lbrace\omega_{10}\in\Omega | n_{t,i}^{*} - \alpha_{i}q< 0,n_{t,i+256}^{*} - \alpha_{i+256}q\geq 0	,n_{t,i+512}^{*} - \alpha_{i+512}q< 0	,n_{t,i+768}^{*} - \alpha_{i+768}q \geq 0	\rbrace $          \\ \hline
$\Omega_{11}$ & $\lbrace\omega_{11}\in\Omega | n_{t,i}^{*} - \alpha_{i}q< 0,n_{t,i+256}^{*} - \alpha_{i+256}q \geq 0	,n_{t,i+512}^{*} - \alpha_{i+512}q < 0	,n_{t,i+768}^{*} - \alpha_{i+768}q< 0	\rbrace $          \\ \hline
$\Omega_{12}$ & $\lbrace\omega_{12}\in\Omega | n_{t,i}^{*} - \alpha_{i}q< 0,n_{t,i+256}^{*} - \alpha_{i+256}q< 0	,n_{t,i+512}^{*} - \alpha_{i+512}q \geq 0	,n_{t,i+768}^{*} - \alpha_{i+768}q\geq 0	\rbrace $          \\ \hline
$\Omega_{13}$ & $\lbrace\omega_{13}\in\Omega | n_{t,i}^{*} - \alpha_{i}q < 0,n_{t,i+256}^{*} - \alpha_{i+256}q< 0	,n_{t,i+512}^{*} - \alpha_{i+512}q\geq 0	,n_{t,i+768}^{*} - \alpha_{i+768}q < 0	\rbrace $    \\ \hline
$\Omega_{14}$ & $\lbrace\omega_{14}\in\Omega | n_{t,i}^{*} - \alpha_{i}q< 0,n_{t,i+256}^{*} - \alpha_{i+256}q< 0	,n_{t,i+512}^{*} - \alpha_{i+512}q < 0	,n_{t,i+768}^{*} - \alpha_{i+768}q\geq 0	\rbrace $    \\ \hline
$\Omega_{15}$ & $\lbrace\omega_{15}\in\Omega | n_{t,i}^{*} - \alpha_{i}q< 0,n_{t,i+256}^{*} - \alpha_{i+256}q< 0	,n_{t,i+512}^{*} - \alpha_{i+512}q< 0	,n_{t,i+768}^{*} - \alpha_{i+768}q< 0	\rbrace $             \\ \hline
\end{tabular}
}
\label{Sets of event of bit error}
\end{table}
\begin{table}[t]
\centering
\caption{The details of support $\Omega_k$ for $m=2$}
\begin{tabular}{|c|c|}
\hline
Set   & Set condition                                                                                                                                                                              \\ \hline
$\Omega_0$  & $\lbrace\omega_0\in\Omega | n_{t,i}^{*} - \alpha_{i}q\geq 0,n_{t,i+256}^{*} - \alpha_{i+256}q\geq 0\rbrace $ \\ \hline
$\Omega_1$  & $\lbrace\omega_1\in\Omega | n_{t,i}^{*} - \alpha_{i}q\geq 0,n_{t,i+256}^{*} - \alpha_{i+256}q < 0		\rbrace $   \\ \hline
$\Omega_2$  & $\lbrace\omega_2\in\Omega | n_{t,i}^{*} - \alpha_{i}q < 0,n_{t,i+256}^{*} - \alpha_{i+256}q \geq 0	\rbrace $    \\ \hline
$\Omega_3$  & $\lbrace\omega_3\in\Omega | n_{t,i}^{*} - \alpha_{i}q < 0,n_{t,i+256}^{*} - \alpha_{i+256}q < 0	\rbrace $      \\ \hline
\end{tabular}
\label{Sets of event of bit error m2}
\end{table}
The bit error support $E$ can be partitioned into $m^2$ supports  $E_0$, $E_1$, $\cdots$, and $E_{m^2-1}$  by using the support $\Omega_k$ as follows:
\begin{eqnarray}
E_{k} = \{\epsilon_k | \epsilon_k \in \Omega_k \cap E \}.
\label{bit error support partition}
\end{eqnarray}
It is obvious that $E_{k} \subseteq \Omega_{k}$ for $j=0$, $1$, $\cdots$, $m^2-1$, $E_{i} \cap E_{j}= \emptyset$ if $i \neq j$, and $E = \cup_{i=0}^{m^2-1}E_{i}$. 
Also, $E_k$ is expressed by using $\Omega_k$ as follows:
\begin{eqnarray}	
E_k = \{ \epsilon_k \in \Omega_k  |  T_m \leq \sum_{l=0}^{m-1} (n_{t,i+256l}^{*} - q\alpha_{i+256l})y_{k,l}^{m} \leq 2T_m  \} \label{beta}
\end{eqnarray}
For the convenience of explanation, the inequality in (\ref{beta}) is expressed by using the new variable $\beta \in [ T_m,2T_m]$ as follows:
\begin{eqnarray}
&&T_m \leq \sum_{l=0}^{m-1} (n_{t,i+256l}^{*} - q\alpha_{i+256l})y_{k,l}^{m} \leq 2T_m \nonumber\\ 
	&\Leftrightarrow&
	\sum_{l=0}^{m-1} (n_{t,i+256l}^{*} - q\alpha_{i+256l})y_{k,l}^{m}=\beta. \nonumber \\
	&\Leftrightarrow&
\sum_{l=0}^{m-1}n_{t,i+256l}^{*}y_{k,l}^{m}=q\Big(\sum_{l=0}^{m-1}\alpha_{i+256l}y_{k,l}^{m}\Big)+\beta \nonumber\\
&\Leftrightarrow&
\sum_{l=0}^{m-1}n_{t,i+256l}^{*}y_{k,l}^{m}=qA_i+\beta, 
	\label{change_constraint}
\end{eqnarray}
where $A_i=\sum_{l=0}^{m-1}\alpha_{i+256l}y_{k,l}^{m}$ and $A_{i}$ is fully determined by $n_{t,i}^{*}$, $n_{t,i+256}^{*}$, $n_{t,i+512}^{*}$, and $n_{t,i+768}^{*}$ for $n=1024$ or by $n_{t,i}^{*}$ and $n_{t,i+256}^{*}$ for $n=512$, and $|A_i|<m\alpha_{max}$ where $\alpha_{max}=\lfloor (2nk^2+k+(q-1)/r)/q \rfloor$.
There are two constraints in (\ref{change_constraint}) such that $A_i$ is a finite integer and $\sum_{l=0}^{m-1}n_{t,i+256l}^{*}$ and $\beta$ are congruent modulo $q$. 
Thus, $E_{k}$ can be expressed as union of supports satisfying two constraints on $n^*_{t,i}$ and $\alpha_i$ as follows:
\begin{eqnarray}
	E_{k} &=& \{ \epsilon_k \in \Omega_k | T_m \leq \sum_{l=0}^{m-1} (n_{t,i+256l}^{*} - q\alpha_{i+256l})y_{k,l}^{m} \leq 2T_m  \} \nonumber  \\
	&=&
	\bigcup_{\beta} \{\epsilon_k \in \Omega_k | 
	\sum_{l=0}^{m-1}  n_{t,i+256l}^{*}y_{k,l}^{m}=qA_i+\beta	
	\} \nonumber \\
	&=&
	\bigcup_{j=A_{min}}^{A_{max}}\Big( \bigcup_{\beta}\{\epsilon_k \in \Omega_k | \sum_{l=0}^{m-1}n_{t,i+256l}^{*}y_{k,l}^{m} = jq+\beta\}
	\cap \{\epsilon_k \in \Omega_k | A_{i}=j\}\Big), \nonumber
\end{eqnarray}
where $A_{min}=-m\alpha_{max}$ and $A_{max}=m\alpha_{max}$.

In order to calculate the BER, the occurring probability $\Pr(E)$ of the bit error support $E$ should be calculated.
As mentioned above, since the bit error support $E$ can be disjointly partitioned, $\Pr(E)=\sum_{i=0}^{m^2-1}\Pr(E_i)$.
For the description of simplicity, we first consider the the event $E_0$ of bit error, and it can be expressed as the union of different supports on $j=0$ and $j\neq 0$ as follows:
\begin{eqnarray}
	E_0&=&\bigcup_{j=A_{min}}^{A_{max}} \Big( \bigcup_{\beta}	\{\epsilon_0 \in \Omega_0 | \sum_{l=0}^{m-1}n_{t,i+256l}^{*}y_{0,l}^{m} = jq+\beta\}
	\cap \{\epsilon_0 \in \Omega_0 | A_{i}=j\}\Big) \nonumber \\
	&=& \bigcup_{j:j\neq 0}
	\Big(
	\bigcup_{\beta}\{\epsilon_0 \in \Omega_0 | \sum_{l=0}^{m-1}n_{t,i+256l}^{*} = jq+\beta\}
	\cap \{\epsilon_0 \in \Omega_0| A_{i}=j\}\Big)  \nonumber \\
	&&\cup \Big(
	\bigcup_{\beta}
	\{\epsilon_0 \in \Omega_0| \sum_{l=0}^{m-1}n_{t,i+256l}^{*}  = \beta\}
	\cap \{\epsilon_0 \in \Omega_0 | A_{i}=0\} \Big) \nonumber \\
	&=& E_{0,j \neq 0} \cup E_{0,j=0} ,
	\label{E0}
\end{eqnarray}
where $\textbf{y}_{0}^{m}$ is all-one vector, $E_{0,j \neq 0}=\bigcup_{j:j\neq 0}\Big(\bigcup_{\beta}\{\epsilon_0 \in \Omega_0|\sum_{l=0}^{m-1}n_{t,i+256l}^{*} = jq+\beta\}\cap \{\epsilon_0 \in \Omega_0| A_{i}=j\}\Big)$, and $E_{0,j=0}=\bigcup_{\beta}	\{\epsilon_0 \in \Omega_0| \sum_{l=0}^{m-1}n_{t,i+256l}^{*}  = \beta\}	\cap \{\epsilon_0 \in \Omega_0 | A_{i}=0\}$.
However, it is difficult to know the exact supports of $E_{0,j \neq 0}$ and $E_{0,j = 0}$, and even if they are correctly known, it is very difficult to calculate the exact occurring probabilities.
Therefore, we derive the upper bounds on the occurring probabilities of each support $E_{0,j \neq 0}$ and $E_{0,j=0}$ through Theorems \ref{thm1} and \ref{thm2}, and by using such upper bounds, the occurring probability $\Pr(E_0)$ can be upper bounded.
\newtheorem{thm}{Theorem}
\begin{thm}
The occurring probability $\Pr(E_{0,j \neq 0})$ of $E_{0,j \neq 0}$ in $(\ref{E0})$ is at most $mPr(|n_{t,i}^{*}|>\lfloor{q \over 2}\rfloor)$.
\label{thm1}
\end{thm}
\begin{proof}
If $A_{i}\neq0$, then at least one of $\alpha_{i}$, $\alpha_{i+256}$, $\alpha_{i+512}$, and $\alpha_{i+768}$ is not zero for $n=1024$.
Similarly, for $n=512$, if $A_{i}\neq0$, then at least one of $\alpha_{i}$ and $\alpha_{i+256}$ is not zero.
In the equation $n_{t,i}^{\mu_{i}=1}=n_{t,i}^{*}-\alpha_{i}q+\lfloor{q \over 2}\rfloor$, since $\alpha_i$ makes $n_{t,i}^{\mu_{i}=1}$ be in $[0, q-1]$, $\alpha_{i}=0$ if and only if $| n_{t,i}^{*}| \leq \lfloor{q \over 2}\rfloor$.
Conversely, $\alpha_{i}\neq 0$ if and only if  $|n_{t,i}^{*}|>\lfloor{q \over 2}\rfloor$.
Therefore, at least one among $|n_{t,i}^{*}|$, $|n_{t,i+256}^{*}|$, $|n_{t,i+512}^{*}|$, and $|n_{t,i+768}^{*}|$ is greater than $\lfloor{q \over 2}\rfloor$ for $n=1024$.
Similarly, at least one among $|n_{t,i}^{*}|$ and $|n_{t,i+256}^{*}|$ is greater than $\lfloor{q \over 2}\rfloor$ for $n=512$.
Then, we can relax the constraint $\{\Omega_0 | \sum_{l=0}^{m-1}n_{t,i+256l}^{*}  = jq + \beta\}$ and make the superset whose occurring probability is greater than or equal to origin set $E_{0,j \neq 0}$ as follows:
	\begin{eqnarray*}
	E_{0,j \neq 0} &=& \bigcup_{j:j\neq 0} \Big(
	\bigcup_{\beta}\{\epsilon_0 \in \Omega_0 | \sum_{l=0}^{m-1}n_{t,i+256l}^{*} = jq+\beta\}
	\cap \{\epsilon_0 \in \Omega_0| A_{i}=j\}\Big)\\
	&\subseteq& \bigcup_{j:j\neq 0}
	\{\epsilon_0 \in \Omega_0 | A_{i}=j\} \\
	&\subseteq&\bigcup_{l=0}^{m-1}
	\Big\{\epsilon_0 \in \Omega_0 \Big| \vert n_{t,i+256l}^{*}\vert > {q \over 2} \Big\}
	\end{eqnarray*}
The occurring probability of $E_{0,j \neq 0}$ is bounded by using the union bound and the fact that the distributions of $n^*_{t,i}, \forall i \in [0,n-1]$ are identical.
	\begin{eqnarray*}
	\Pr(E_{0,j \neq 0}) 
	&\leq& \Pr \Big( \bigcup_{l=0}^{m-1}\{\vert n_{t,i+256l}^{*}\vert > {q \over 2} \} \Big ) \\
	&\leq& \sum_{l=0}^{m-1}\Pr\Big( \vert n_{t,i+256l}^{*}\vert > {q \over 2} \Big)   \\
	&\leq& m\Pr\Big(\vert n_{t,i}^{*}\vert > {q \over 2} \Big). 
	\end{eqnarray*}  
	\qed
\end{proof}
\begin{figure}[t]
\centering
\includegraphics[width=\textwidth]{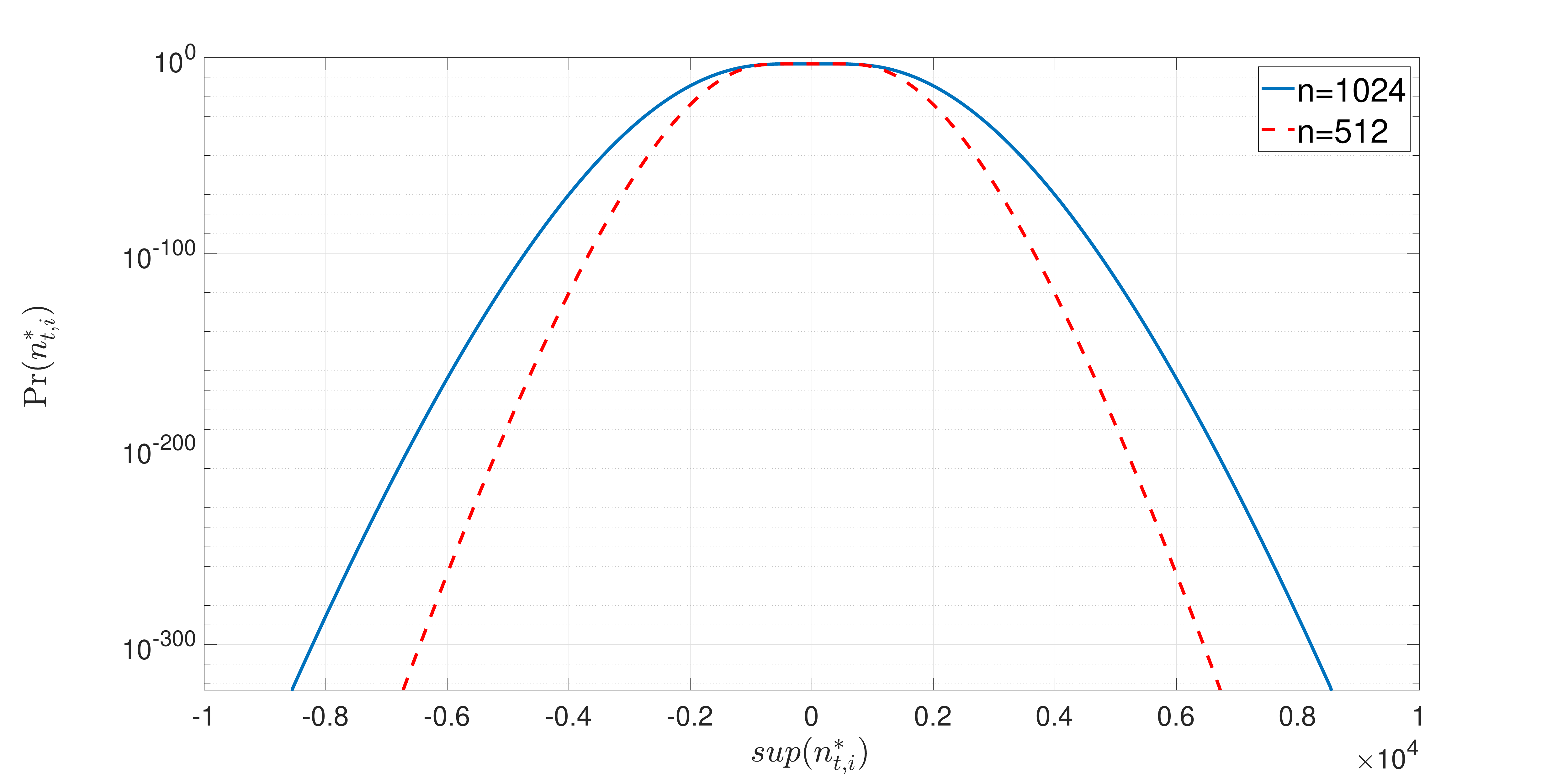}
\caption{The distribution of $n_{t,i}^{*}$ for $n=1024$ ($m=4$) and $n=512$ ($m=2$).} 
\label{n_t*}
\end{figure}
The distribution of $n_{t,i}^{*}$ can be numerically calculated as shown in Fig. \ref{n_t*}. 
By using the distribution of $n_{t,i}^{*}$, we can calculate $\Pr(E_{0,j \neq 0}) \leq 2^{-564}$ for $n=1024$ and $\Pr(E_{0,j \neq 0}) \leq 2^{-908}$ for $n=512$.
\begin{thm}
The occurring probability of $E_{0,j=0}$ is at most $\Pr(T_m \leq \sum_{l=0}^{m-1}n_{t,i+256l}^{*} \leq 2T_m)$.
\label{thm2}
\end{thm}
\begin{proof}
If $A_{i}=0$, then the superset of support of $E_{0,j=0}$ can be found by relaxing the constraints $\{\epsilon_0 \in \Omega_0 | A_{i}=0\}$ as follows:
\begin{eqnarray*}
	E_{0,j=0} &=&\bigcup_{\beta}\{\epsilon_0 \in \Omega_0| \sum_{l=0}^{m-1}n_{t,i+256l}^{*}  = \beta\}	\cap \{\epsilon_0 \in \Omega_0 | A_{i}=0\}\\
	&\subseteq&\bigcup_{\beta}
	\{\epsilon_0 \in \Omega_0 | \sum_{l=0}^{m-1}n_{t,i+256l}^{*}  = \beta\}.
\end{eqnarray*}
Therefore, the occurring probability of $E_{0,j=0}$ can be upper bounded through the union bound as follows:
\begin{eqnarray}
\Pr(E_{0,j=0}) &=& \Pr(\bigcup_{\beta}
	\{\epsilon_0 \in \Omega_0 | \sum_{l=0}^{m-1}n_{t,i+256l}^{*}  = \beta\}) \nonumber \\
&\leq& \Pr(T_m \leq \sum_{l=0}^{m-1}n_{t,i+256l}^{*} \leq 2T_m). \nonumber
\end{eqnarray}
\qed
\end{proof}
To calculate the upper bound of probability occurring $\Pr(E_0)$ through Theorems 1 and 2, the distributions of $\sum_{l=0}^{m-1}n_{t,i+256l}^{*}$ are required.
However, since $ n_{t,i}^{*} $, $ n_{t,i+256}^{*} $, $ n_{t,i+512}^{*} $, and $ n_{t,i+768}^{*}$ for $n=1024$ or $ n_{t,i}^{*} $ and $ n_{t,i+256}^{*} $ for $n=512$ are statistically dependent to each other, it is not only difficult to analytically calculate the occurring probability $\Pr(E_{0,j=0})$, but also not numerically computable.
However, in this paper, the distribution of $\sum_{l=0}^{m-1}n_{t,i+256l}^{*}$ can be numerically computable by decomposing $\sum_{l=0}^{m-1}n_{t,i+256l}^{*}$ into the sum of \textit{i.i.d.} random variables through the following Theorem \ref{thm3}.
\begin{thm}
\label{thm3}   
$\sum_{l=0}^{m-1} n_{t,i+256l}^{*}$ is decomposed into the sum of i.i.d random variables.
\end{thm} 
\begin{proof}
We know that $n_{t,i}^{*}=(\textbf{e} \odot \textbf{s}')_{i}-(\textbf{e}' \odot \textbf{s})_{i}+\textbf{e}''_{i}+\textbf{n}_{c,i}$. Then, 
  \begin{eqnarray}	
	\sum_{l=0}^{m-1} n_{t,i+256l}^{*} 
	&=& \sum_{l=0}^{m-1} (\textbf{e} \odot \textbf{s}')_{i+256l}
	- \sum_{l=0}^{m-1} (\textbf{e}' \odot \textbf{s})_{i+256l} 
	+ \sum_{l=0}^{m-1} \textbf{e}''_{i+256l}
	+ \sum_{l=0}^{m-1} \textbf{n}_{c,i+256l} \nonumber \\
	\label{thm3_1}
  \end{eqnarray}
If $\sum_{l=0}^{m-1} (\textbf{e} \odot \textbf{s})_{i+256l}$ is decomposed into \textit{i.i.d.} random vectors, $\sum_{l=0}^{m-1} n_{t,i+256l}^{*}$ can be also decomposed into \textit{i.i.d.} random vectors. 
An inner product of two vectors can be decomposed into the sum of inner products of sub vectors.
Thus, $\sum_{l=0}^{m-1} (\textbf{e} \odot \textbf{s})_{i+256l}$ can be decomposed into the sum of inner products of sub vectors as follows for $n=1024$:
   \begin{eqnarray*}
   \sum_{l=0}^3 (\textbf{e} \odot \textbf{s})_{i+256l}&=&
   \begin{pmatrix}
     e_{0}\\
     \vdots \\
     e_{256}\\
   \vdots \\     
     e_{512}\\
   \vdots \\     
     e_{768}\\
     \vdots \\
   \end{pmatrix}
   ^{T}
   \begin{pmatrix}
     \begin{pmatrix}
     +s_{0}\\
     \vdots \\
     -s_{768}\\
   \vdots \\     
     -s_{512}\\
   \vdots \\     
     -s_{256}\\
     \vdots \\
   \end{pmatrix}   
   +
   \begin{pmatrix}
     +s_{256}\\
     \vdots \\
     +s_{0}\\
   \vdots \\     
     -s_{768}\\
   \vdots \\     
     -s_{512}\\
     \vdots \\
   \end{pmatrix}   
   +
   \begin{pmatrix}
     +s_{512}\\
     \vdots \\
     +s_{256}\\
   \vdots \\     
     +s_{0}\\
   \vdots \\     
     -s_{768}\\
     \vdots \\
   \end{pmatrix}   
   +\begin{pmatrix}
     +s_{768}\\
     \vdots \\
     +s_{512}\\
   \vdots \\     
     +s_{256}\\
   \vdots \\     
     +s_{0}\\
     \vdots \\
   \end{pmatrix}      
   \end{pmatrix}   \\
   &=&
   \begin{pmatrix}
     e_{0}\\
     e_{256}\\
     e_{512}\\   
     e_{768}\\
   \end{pmatrix}
   ^{T}
   \begin{pmatrix}
     \begin{pmatrix}
     +s_{0}\\
     -s_{768}\\   
     -s_{512}\\   
     -s_{256}\\
   \end{pmatrix}   
   +
   \begin{pmatrix}
     +s_{256}\\
     +s_{0}\\   
     -s_{768}\\   
     -s_{512}\\
   \end{pmatrix}   
   +
   \begin{pmatrix}
     +s_{512}\\
     +s_{256}\\    
     +s_{0}\\   
     -s_{768}\\
   \end{pmatrix}   
   +\begin{pmatrix}
     +s_{768}\\
     +s_{512}\\
     +s_{256}\\   
     +s_{0}\\
   \end{pmatrix}      
   \end{pmatrix}\\
   &+&
   \begin{pmatrix}
     e_{1}\\
     e_{257}\\
     e_{513}\\   
     e_{769}\\
   \end{pmatrix}
   ^{T}
   \begin{pmatrix}
     \begin{pmatrix}
     +s_{1}\\
     -s_{769}\\   
     -s_{513}\\   
     -s_{257}\\
   \end{pmatrix}   
   +
   \begin{pmatrix}
     +s_{257}\\
     +s_{1}\\   
     -s_{769}\\   
     -s_{513}\\
   \end{pmatrix}   
   +
   \begin{pmatrix}
     +s_{513}\\
     +s_{257}\\    
     +s_{1}\\   
     -s_{769}\\
   \end{pmatrix}   
   +\begin{pmatrix}
     +s_{769}\\
     +s_{513}\\
     +s_{257}\\   
     +s_{1}\\
   \end{pmatrix}      
   \end{pmatrix}\\
   &\qquad & \qquad\qquad\qquad\qquad\qquad\qquad\qquad\vdots\\
   &+&
   \begin{pmatrix}
     e_{255}\\
     e_{511}\\
     e_{767}\\   
     e_{1023}\\
   \end{pmatrix}
   ^{T}
   \begin{pmatrix}
     \begin{pmatrix}
     +s_{255}\\
     -s_{1023}\\   
     -s_{767}\\   
     -s_{511}\\
   \end{pmatrix}   
   +
   \begin{pmatrix}
     +s_{511}\\
     +s_{255}\\   
     -s_{1023}\\   
     -s_{767}\\
   \end{pmatrix}   
   +
   \begin{pmatrix}
     +s_{767}\\
     +s_{511}\\    
     +s_{255}\\   
     -s_{1023}\\
   \end{pmatrix}   
   +\begin{pmatrix}
     +s_{1023}\\
     +s_{767}\\
     +s_{511}\\   
     +s_{255}\\
   \end{pmatrix}      
   \end{pmatrix}\\
   \end{eqnarray*}
It is clear that each inner product of sub vectors is a similar structure and hence, we define new random variable $W_j$ for $n=1024$, 
      \begin{eqnarray}
   W_{j}=
   \begin{pmatrix}
     e_{j}\\
     e_{j+256}\\
   e_{j+512}\\
   e_{j+768}
     \end{pmatrix} ^T
   \begin{pmatrix}
     \begin{pmatrix}
     +s_{j}\\
     -s_{j+768}\\
     -s_{j+512}\\
     -s_{j+256}\\
   \end{pmatrix}   
   +
   \begin{pmatrix}
     +s_{j+256}\\
     +s_{j}\\    
     -s_{j+768}\\
     -s_{j+512}\\
   \end{pmatrix}   
   +
   \begin{pmatrix}
     +s_{j+512}\\
     +s_{j+256}\\
     +s_{j}\\
     -s_{j+768}\\
   \end{pmatrix}   
   +\begin{pmatrix}
     +s_{j+768}\\
     +s_{j+512}\\
     +s_{j+256}\\
     +s_{j}\\ 
   \end{pmatrix}      
   \end{pmatrix},  \nonumber \\
   \label{W} 
   \end{eqnarray}
and for $n=512$,
\begin{eqnarray}
W_{j}=
   \begin{pmatrix}
     e_{j}\\
     e_{j+256} 
   \end{pmatrix} ^T
   \begin{pmatrix}
     \begin{pmatrix}
     +s_{j}\\
     -s_{j+256}\\
     \end{pmatrix}   
   +
     \begin{pmatrix}
     +s_{j+256}\\
     +s_{j}\\    
     \end{pmatrix}   
   \end{pmatrix},  \nonumber \\
   \label{W_m2} 
   \end{eqnarray}  
and $\sum_{l=0}^{m-1} (\textbf{e} \odot \textbf{s})_{i+256l}=\sum_{j=0}^{255}W_j$.
Since $W_j$ and $W_{j'}$ for $j\neq j'$ consist of different random variables, $W_j$'s are clearly independent to each other.
Thus,  
   \begin{eqnarray*}   
   \sum_{l=0}^{m-1} n_{t,i+256l}^{*} 
   &=& 
   \sum_{l=0}^{m-1} \Big((\textbf{e}' \odot \textbf{s})_{i+256l}-(\textbf{e} \odot \textbf{s}')_{i+256l} + e''_{i+256l} + n_{c,i+256l}\Big)
   \\
   &=&
   \sum_{j=0}^{511}W_{j}+\sum_{l=0}^{m-1} \Big(e''_{i+256l} +n_{c,i+256l}\Big).
   \end{eqnarray*}
Note that $(\textbf{e}' \odot \textbf{s})$ and $(\textbf{e} \odot \textbf{s}')$ are decomposed into $256$ \textit{i.i.d.} random variables $W_j$, respectively.
Therefore, $n_{t,i}^{*}$ can be decomposed into $512$ random variables $W_j$ and $2m$ random variables of $e''$ and $n_c$. \qed
\end{proof}
It is difficult to calculate the distribution of $\sum_{l=0}^{m-1}n_{t,i+256l}^{*}$ since $\sum_{l=0}^{m-1}n_{t,i+256l}^{*}$ consists of the products and sums of $ 4n + 2m $ random variables.
However, since $\sum_{l=0}^{m-1}n_{t,i+256l}^{*}$ is converted into the sum of $ 512 + 2m $ random variables, the distribution of $\sum_{l=0}^{m-1}n_{t,i+256l}^{*}$ becomes numerically computable.

In conclusion, by using Theorems \ref{thm1}, \ref{thm2}, and the union bound, the occurring probability of $E_0$ is upper bounded as follows:
\begin{eqnarray*}
\Pr(E_0) &=&  \Pr(E_{0,j \neq 0} \cup E_{0,j = 0})\\
&\leq& \Pr(E_{0,j \neq 0}) + \Pr(E_{0,j = 0})\\ 
&\leq& m\Pr\Big(\vert n_{t,i}^{*}\vert > {q \over 2} \Big) + \Pr\Big(T_m \leq \sum_{l=0}^{m-1}n_{t,i+256l}^{*} \leq 2T_m\Big).
\end{eqnarray*}

Next, in order to calculate the BER, $\Pr(E_1)$, $\Pr(E_2)$, $\cdots$, and $\Pr(E_{15})$ should be calculated, and they can be calculated by using following Theorem \ref{thm4}.	
\begin{thm}
 $\Pr(E_{k})\leq m\Pr(\vert n_{t,i}^{*}\vert > {q \over 2}) + \Pr(T_m\leq \sum_{l=0}^{m-1}n_{t,i+256l}^{*} \leq 2T_m), \forall k \in [0,m^2-1]$. 
\label{thm4}
\end{thm}
\begin{proof}
$E_k$ can also expressed as a union of $E_{k,j\neq0}$ and $E_{k,j=0}$ by using $\textbf{y}^m_k$, similar to (\ref{E0}).
First, we consider $E_{k,j\neq0}$ and then likewise the proof of Theorem \ref{thm1}, the superset of $E_{k,j \neq 0}$ can be found.
If $\sum_{l=0}^{m-1} \alpha_{i+256l}  y^{m}_{k,l} \neq 0$, then at least one among $\alpha_i$, $\alpha_{i+256}$, $\alpha_{i+512}$, and $\alpha_{i+768}$ is not zero for $n=1024$.  
Similarly, for $n=512$, if $\sum_{l=0}^{m-1} \alpha_{i+256l}  y^{m}_{k,l} \neq 0$, then at least one among $\alpha_i$ and $\alpha_{i+256}$ is not zero.
The fact implies at least one among $|n_{t,i}^{*}|$, $|n_{t,i+256}^{*}|$, $|n_{t,i+512}^{*}|$, and $|n_{t,i+768}^{*}|$ is greater than $\lfloor{q / 2}\rfloor$.
Then, we can also relax the constraint $\{\Omega_k | \sum_{l=0}^{m-1}n_{t,i+256l}^{*}y^{m}_{k,l}  = jq + \beta\}$ and make the superset whose occurring probability is greater than or equal to origin set $E_{k,j \neq 0}$ likewise $E_{0,j\neq0}$ as follows:
    \begin{eqnarray*}
	E_{k,j \neq 0} &=& \bigcup_{j:j\neq0} \Big(\bigcup_{\beta} \{\epsilon_k \in \Omega_k| \sum_{l=0}^{m-1}n_{t,i+256l}^{*}y^{m}_{k,l} = jq+\beta\}\cap \{\epsilon_k \in \Omega_k | A_{i}=j\}\Big)\\
	&\subseteq& \bigcup_{j:j\neq0}
	\{\epsilon_k \in \Omega_k| A_{i}=j\} \\
	&=&	\bigcup_{j:j\neq0}
	\{\epsilon_k \in \Omega_k | \sum_{l=0}^{m-1} \alpha_{i+256l} y^{m}_{k,l}=j\} \\
	&\subseteq& \bigcup_{l=0}^{m-1} 
	\Big\{\epsilon_k \in \Omega_k \Big| \vert n_{t,i+256l}^{*}\vert > {q \over 2} \Big\}
	\end{eqnarray*}
Clearly,  $\Pr(E_{k,j \neq 0})$, $\forall k \in [1,m^2-1]$ is upper bonded as same as $\Pr(E_{0,j \neq 0})$ by using the union bound as follows:
	\begin{eqnarray*}
	\Pr(E_{k,j \neq 0}) 
	&\leq& \Pr \Big( \bigcup_{l=0}^{m-1}\big\{\vert n_{t,i+256l}^{*}\vert > {q \over 2} \big\} \Big ) \\
	&\leq& \sum_{l=0}^{m-1}\Pr\Big( \vert n_{t,i+256l}^{*}\vert > {q \over 2} \Big)   \\
	&\leq& m\Pr\Big(\vert n_{t,i}^{*}\vert > {q \over 2} \Big).
	\end{eqnarray*}
      Also, Theorem \ref{thm2} is applied to other $E_{k,j = 0}$, $\forall k \in [1,m^2-1]$ as follows:
  	\begin{eqnarray*}
	E_{k,j=0} &=&\Big(\bigcup_{\beta}\{\epsilon_k \in \Omega_k| \sum_{l=0}^{m-1} n_{t,i+256l}^{*}y^{m}_{k,l}  = \beta\}	\cap \{\epsilon_k \in \Omega_k | A_{i}=0\} \Big)\\
	&\subseteq&\bigcup_{\beta} \{\epsilon_k \in \Omega_k | \sum_{l=0}^{m-1}n_{t,i+256l}^{*}y^{m}_{k,l} = \beta\} .
	\end{eqnarray*}
Therefore, we obtain the upper bound on $\Pr(E_{k,j=0})$, $\forall k \in [1,m^2-1]$ as follows:
   \begin{eqnarray*}
    \Pr(E_{k,j= 0}) \leq \Pr\Big(T_m \leq \sum_{l=0}^{m-1}n_{t,i+256l}^{*} y^{m}_{k,l} \leq 2T_m \Big).
   \end{eqnarray*}
Since expectation of $W_j$ in (\ref{W}) for $j=0$, $1$, $\cdots$, $511$  is sum of product of \textit{i.i.d.} random variables of $e$, $e'$, $s$, $s'$, and $s''$ whose means are zero, the expectation of $W_j$ is zero. 
Also, since the distributions of $e$, $e'$, $s$, $s'$, and $s''$ are symmetric, the distribution of $W_j$ is symmetric.
This fact guarantees that for any $\textbf{y}^m_k$, the distributions of $\sum_{l=0}^{m-1}n_{t,i+256l}^{*}y^{m}_{k,l}$ are statistically identical and therefore the upper bounds on $\Pr(E_{1})$, $\Pr(E_{2})$, $\cdots$, and $\Pr(E_{m^2-1})$ are same as $\Pr(E_{0})$. \qed
\end{proof}

In summary, by using Theorems \ref{thm1}, \ref{thm2}, \ref{thm3}, and \ref{thm4}, the upper bound on $\Pr(E)$, which is the BER of NewHope, is derived by using the union bound as follows:
\begin{eqnarray}
\Pr(E) &=&\Pr\Big(\bigcup_{k=0}^{m^2-1}E_k\Big) \nonumber \\
&\leq& \sum_{k=0}^{m^2-1}\Pr(E_k) \nonumber \\
 &\leq& m^2\Bigg( m\Pr\Big(\vert n_{t,i}^{*}\vert > {q \over 2} \Big) + \Pr\Big(T_m \leq \sum_{l=0}^{m-1}n_{t,i+256l}^{*} \leq 2T_m\Big) \Bigg). \nonumber \\
\label{BER upper bound}
\end{eqnarray}
\subsection{Derivation of Upper Bound on DFR of NewHope}
By using the $\Pr(E)$ in (\ref{BER upper bound}), the DFR can be easily upper bounded by using the union bound.
\begin{thm} The DFR $\Pr(\mu\neq\mu')$ of NewHope is upper bounded as 
$\Pr(\mu\neq\mu')\leq \sum_{i=0}^{255} \Pr(\mu_i \neq \mu'_{i})$.
\end{thm}
\begin{proof} Since the DFR is the union of all bit error events, the DFR is upper bounded by the sum of BERs by using the union bound as follows:
\begin{eqnarray*}
			DFR &=& \Pr\Bigg(\bigcup_{i=0}^{255}(\mu_i \neq \mu'_{i})\Bigg) \nonumber \\
			&\leq & \sum_{i=0}^{255} \Pr(\mu_i \neq \mu'_{i})\\
			&=& 256\Pr(\mu_i \neq \mu'_{i}).
\end{eqnarray*}	
\end{proof}
\qed
Each BER of outputs of ATE decoder is identical so that the upper bound on the DFR is expressed as:
\begin{eqnarray}
DFR\leq 256m^2\Bigg( m\Pr\Big(\vert n_{t,i}^{*}\vert > {q \over 2} \Big) + \Pr\Big(T_m \leq \sum_{l=0}^{m-1}n_{t,i+256l}^{*} \leq 2T_m \Big) \Bigg). \nonumber \\
\label{DFR}
\end{eqnarray}

\subsection{Parametrization of the Proposed Upper Bound on DFR of NewHope}
The computational complexity of deriving the distribution of $\sum_{l=0}^{m-1}n_{t,i+256l}^{*}$ is $O(k^{2m})$ since $k^{2m}$ operations are required to calculate the distribution of $W_j$.  
Therefore, as $k$ increases, the proposed upper bound on DFR of NewHope cannot be easily computed.
For this reason, the proposed upper bound on DFR of NewHope is parametrized for easy calculation by using CC bound in spite of losing some tightness.
\begin{thm}[Chernoff-Cramer bound]
     Let $\Phi$ be a distribution over $\mathbb{R}$
      and let $\chi_0,...,\chi_{n-1}$ be \textit{i.i.d.} random variable of $\Phi$,
      with average $\mu$. Then, for any t such that $M_{\Phi_{\chi}}(t)= E_{\chi}[\exp(\chi t)] < \infty$
      it holds that
      \begin{equation}
      \Pr\Big[\sum_{i=0}^{n-1} \chi_{i}>n\mu+\beta\Big]\leq \inf_{t}
      \exp(\beta t +n\ln[M_{\Phi_{\chi}}(t)])
      \end{equation}
   \end{thm}
   The proposed upper bound on DFR of NewHope is the sum of two occurring probabilities $\Pr(\vert n_{t,i}^{*}\vert > {q \over 2})$ and $\Pr(T_m\leq \sum_{l=0}^{m-1}n_{t,i+256l}^{*} \leq 2T_m)$ in (\ref{DFR}) and those probabilities can be parameterized with CC bound, respectively. 
In order to apply CC bound to $\Pr(\vert n_{t,i}^{*}\vert > {q \over 2})$, we need to calculate the moment generating function (MGF) of product of two random variables following the centered binomial distribution.
Suppose that $X$ and $Y$ follow the binomial distribution with parameter $2k$, and  $X_c$ and $Y_c$ follow the centered binomial distribution with parameter $k$. Then $X_c = X - k$ and $Y_c = Y - k$ and the MGF $M_{\Phi_{X_c \cdot Y_c}}(t)$ of $X_c \cdot Y_c$ is calculated as follows:
   \begin{eqnarray}
   M_{\Phi_{X_c \cdot Y_c}}(t)&=&E_{X,Y}\Big[e^{(x-k)(y-k)t}\Big] \nonumber \\
   &=&E_{Y}\Big[E_{X}[e^{(x-k)(y-k)t}]\Big] \nonumber \\
   &=&E_{Y}\Big[\sum_{y=0}^{2k}
   {\binom{2k}{x}}e^{(x-k)(y-k)t}2^{-x}2^{-(2k-x)}\Big]\nonumber \\
   &=&E_{Y}\Big[\sum_{y=0}^{2k}
   {\binom{2k}{x}}e^{-kt(y-k)}\Big({ 1\over 2} e^{t(y-k) }\Big)^{x}2^{-(2k-x)}\Big] \nonumber \\
   &=&E_{Y}\Big[ e^{-kt(Y-k)}\Big({ 1\over 2} (e^{t(y-k)}+1)\Big)^{2k}\Big]\nonumber \\
   &=&E_{Y}\Big[\cosh^{2K}\Big({t(y-k) \over 2}\Big)\Big] \nonumber \\
   &=&E_{Y_c}\Big[\cosh^{2K}\Big({ty_c \over 2}\Big)\Big].
   \label{mgf_first}
   \end{eqnarray}
Since $ n_{t,i}^{*}$ is the sum of products of two \textit{i.i.d.} random variables drawn from the centered binomial distribution, CC bound can be applied as follows:
   \begin{eqnarray*}
   \Pr\Big(\vert n_{t,i}^{*}\vert > {q \over 2}\Big)
   &=&\Pr\Big(n_{t,i}^{*} > {q \over 2}\Big)+\Pr\Big( n_{t,i}^{*} < -{q \over 2}\Big) \\
   &=&2\Pr\Big((\textbf{e} \odot \textbf{s}')_{i} -(\textbf{e}' \odot \textbf{s})_{i} > {q \over 2}-(e''_i-n_{c,i})\Big) \\
   &\leq& 2\Pr\Big((\textbf{e} \odot \textbf{s}')_{i} -(\textbf{e}' \odot \textbf{s})_{i}  > {q \over 2}-\Big(k+{q-1 \over 2r}\Big)\Big)\\
   &\leq& \inf_{t}2\exp\Big({q \over 2}-\Big(k+{q-1 \over 2r}\Big)t + 2n\ln E_{Y}\Big[\cosh^{2K}\Big({ty_c \over 2}\Big)\Big]\Big).
   \end{eqnarray*}
Although the MGF of $\sum_{l=0}^{m-1}n_{t,i+256l}^{*}$ is very complicated, Theorem \ref{thm3} guarantees that  $\sum_{l=0}^{m-1}n_{t,i+256l}^{*}$ can be decomposed into \textit{i.i.d} random variables $W_j$ such as $   \sum_{l=0}^{m-1} n_{t,i+256l}^{*} =
   \sum_{j=0}^{511}W_{j}+\sum_{l=0}^{m-1} [e''_{i+256l} +n_{c,i+256l}]$, where $W_j$ is in (\ref{W}).   
For the convenience of analysis, the new variable $W$ is defined as:
 \begin{eqnarray*}
   W=
   \begin{pmatrix}
     e_{0}\\
     e_{1}\\
   e_{2}\\
   e_{3}
     \end{pmatrix}^{T}
   \begin{pmatrix}
     \begin{pmatrix}
     +s_{0}\\
     -s_{3}\\
     -s_{2}\\
     -s_{1}\\
   \end{pmatrix}   
   +
   \begin{pmatrix}
     +s_{1}\\
     +s_{0}\\    
     -s_{3}\\
     -s_{2}\\
   \end{pmatrix}   
   +
   \begin{pmatrix}
     +s_{2}\\
     +s_{1}\\
     +s_{0}\\
     -s_{3}\\
   \end{pmatrix}   
   +\begin{pmatrix}
     +s_{3}\\
     +s_{2}\\
     +s_{1}\\
     +s_{0}\\
   \end{pmatrix}      
   \end{pmatrix}.
   \end{eqnarray*}
The MGF $M_{\Phi_W}(t)$ of $W$ is 
	\begin{eqnarray}
    M_{\Phi_W}(t)=E_{s_0,s_1,s_2,s_3}&\Bigg[&E_{e_0}\Big[\exp(e_0(s_0+s_1+s_2+s_3)t) \Big] \nonumber \\
    &\cdot& E_{e_1}\Big[\exp(e_1(s_0+s_1+s_2-s_3)t) \Big] \nonumber \\
    &\cdot& E_{e_2}\Big[\exp(e_2(s_0+s_1-s_2-s_3)t) \Big] \nonumber \\
    &\cdot& E_{e_3}\Big[\exp(e_3(s_0-s_1-s_2-s_3)t) \Big]
    \Bigg].
    \end{eqnarray}
By using $M_{\Phi_{X_c \cdot Y_c}}(t)=E_{Y_c}[E_{X_c}[\exp(x_cy_ct)]]=E_{Y_c}[\cosh^{2K}({ty_c \over 2})]$ in (\ref{mgf_first}), 
   \begin{eqnarray}
   M_{\Phi_W}(t)=E_{s_0,s_1,s_2,s_3}\Bigg[
   		 &\cosh^{2k}&\Big({t \over 2}(s_0+s_1+s_2+s_3)\Big) \nonumber\\
   \cdot &\cosh^{2k}&\Big({t \over 2}(s_0+s_1+s_2-s_3)\Big) \nonumber\\
   \cdot &\cosh^{2k}&\Big({t \over 2}(s_0+s_1-s_2-s_3)\Big) \nonumber\\
   \cdot &\cosh^{2k}&\Big({t \over 2}(s_0-s_1-s_2-s_3)\Big)\Bigg].
   \end{eqnarray}
Even if the computational complexity of $M_{\Phi_W}$ is $O(k^{2m})$,  by using $\cosh^{2k}(t)\leq e^{-{kt^2}}$  and new random variable $Z=(s_{0}+s_{1}+s_{2}+s_{3})^2+(s_{0}+s_{1}+s_{2}-s_{3})^2+(s_{0}+s_{1}-s_{2}-s_{3})^2+(s_{0}-s_{1}-s_{2}-s_{3})^2  $, the upper bound on $M_{\Phi_W}$ can be derived, which has the complexity $O(k^{m})$ as follows:     
   \begin{align}
   M_{\Phi_W}(t)&\leq
   E_{s_0,s_1,s_2,s_3}\Bigg[
         \exp\Big({kt^2 \over 4}(s_0+s_1+s_2+s_3)^2\Big) \nonumber\\
   &\quad \quad \quad \quad \quad \quad \cdot \exp\Big({kt^2 \over 4}(s_0+s_1+s_2-s_3)^2\Big)\nonumber\\
   &\quad \quad \quad \quad \quad \quad\cdot \exp\Big({kt^2 \over 4}(s_0+s_1-s_2-s_3)^2\Big)\nonumber\\
   &\quad \quad \quad \quad \quad \quad\cdot \exp\Big({kt^2 \over 4}(s_0-s_1-s_2-s_3)^2\Big)\Bigg]\nonumber \\
   &\leq E_Z\Big[\exp\Big({zkt^2 \over 4}\Big)\Big].		
   \label{moment bound}
	\end{align}
Then, by using CC bound and (\ref{moment bound}), $\Pr(T_m\leq \sum_{l=0}^{m-1}n_{t,i+256l}^{*} \leq 2T_m)$ is upper bounded as follows:
   \begin{eqnarray*}
   &&\Pr\Bigg(T_m\leq \sum_{l=0}^{m-1}n_{t,i+256l}^{*} \leq 2T_m\Bigg)\\
   &\leq&\Pr\Bigg(\sum_{i=0}^{511}W_{i} +\sum_{j=0}^{m-1}(e''_j+n_{c,j}) > T_m\Bigg)\\
   &\leq&\Pr\Bigg( \sum_{i=0}^{511}W_{i} \geq T_m - m\Big(k + {q-1 \over 2r}\Big)\Bigg)\\
   &\leq& \inf_{t}\exp\Bigg\{\Bigg(T_m-m\Big(k + {q-1 \over 2r}\Big)\Bigg)t + 512\ln M_{\Phi_W}(t)\Bigg\} \\
   &\leq& \inf_{t}\exp\Bigg\{\Bigg(T_m-m\Big(k + {q-1 \over 2r}\Big)\Bigg)t + 512\ln E_{Z}\Big[\exp\Big({zkt^2\over4}\Big)\Big]\Bigg\}.
   \end{eqnarray*}
Finally, a simplified upper bound on DFR of NewHope is derived as follows:
   \begin{eqnarray}
   DFR&\leq& 256m^2 \Bigg(\inf_{t}\exp\Bigg\{\Bigg(T_m-m\Big(k + {q-1 \over 2r}\Big)\Bigg)t + 512\ln E_{Z}\Big[\exp\Big({zkt^2\over4}\Big)\Big]\Bigg\} \nonumber \\
      &&+m\inf_{t}\exp\Bigg\{\Bigg(T_m-m\Big(k + {q-1 \over 2r}\Big)\Bigg)t + 2n\ln E_{Y}\Big[\cosh ^{2k}\Big({ty\over 2}\Big)\Big] \Bigg\} \Bigg). \nonumber \\
      \label{CC bound}
   \end{eqnarray}

\subsection{Verification of the Proposed Upper Bounds on DFR of NewHope}
We compare the proposed upper bound in (\ref{DFR}) and the simplified upper bound using CC bound in (\ref{CC bound}) with the current upper bound on DFR of NewHope \cite{NewHope NIST}, \cite{NewHope1} for various $k$.
Note that the current upper bound on DFR of NewHope \cite{NewHope NIST}, \cite{NewHope1} is only provided when $k=8$.
Additionally, we compare the proposed upper bounds with the DFR derived by assuming no error dependency as in \cite{Fritzmann}.
For convenience of expression, we will use "Proposed upper bound" to denote the the upper bound derived in (\ref{DFR}), "CC upper bound" to denote the simplified upper bound using CC bound in (\ref{CC bound}), "Current upper bound" to denote the current upper bound on DFR of NewHope \cite{NewHope NIST}, \cite{NewHope1}, "No error dependency" to denote the DFR values calculated by assuming no error dependency as in \cite{Fritzmann}, and "Monte Carlo" to denote the DFR values obtained by performing Monte Carlo simulation of NewHope protocol.
\begin{figure}[t]
\centering
\includegraphics[width=\textwidth]{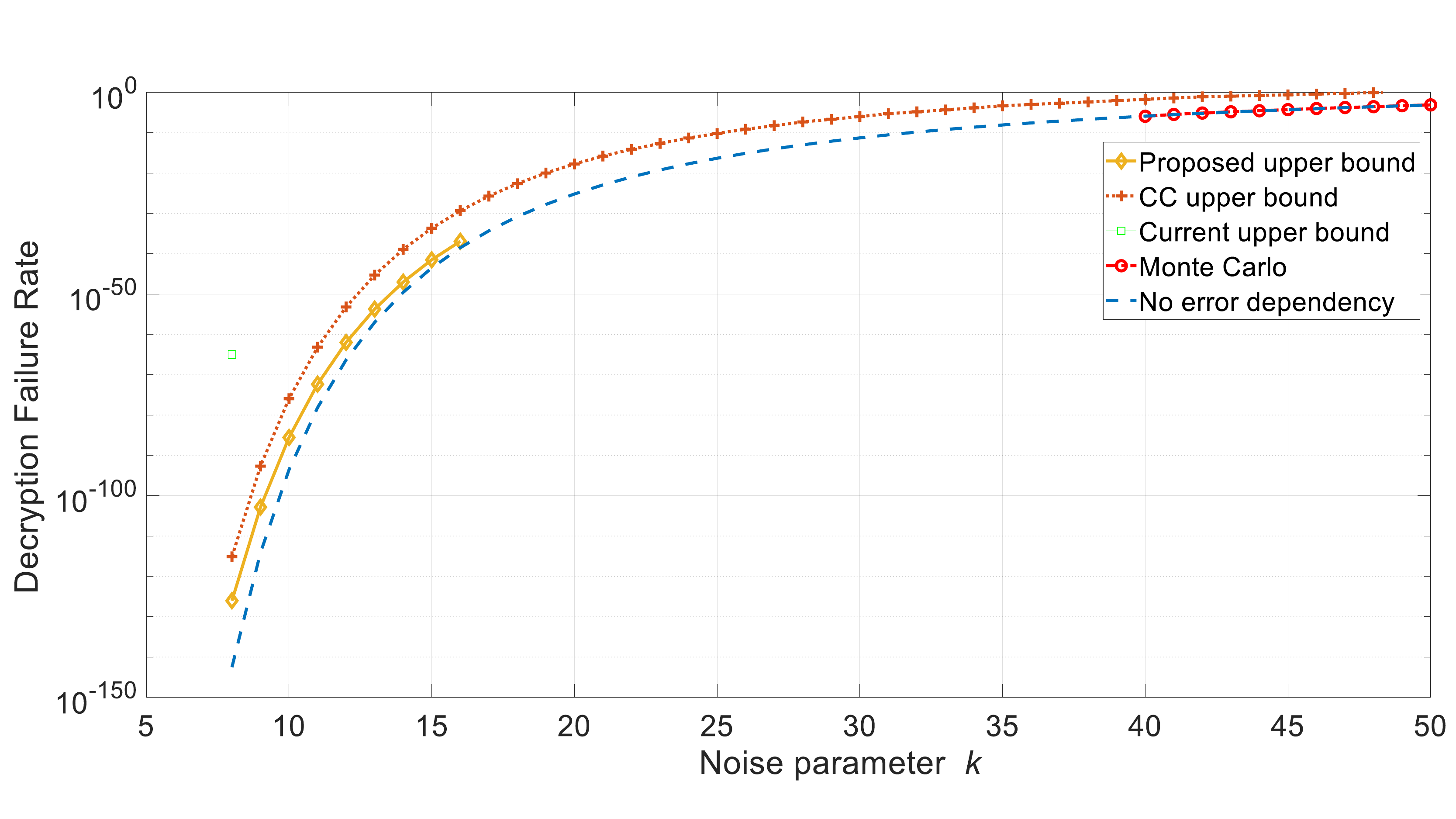}
\caption{Comparison of various upper bounds for various $k$ ($n=1024$). } 
\label{DFR_n1024}
\end{figure}

Fig. \ref{DFR_n1024} compares the various upper bounds on DFR of NewHope for various noise parameter $k$ for $n=1024$.
First of all, it is confirmed that the two proposed upper bounds improve the upper bound more than fifty order of magnitude compared to the current upper bound for $k=8$.
Note that the proposed upper bound on DFR of NewHope is less than $10^{-126}$, the simplified upper bound is less than $10^{-115}$, and the current upper bound is less than $10^{-64}$.
If we compare the proposed upper bound and CC upper bound, we can see that CC bound is more loose as expected.
Nevertheless, since the computational complexity of the proposed upper bound substantially increases as $k$ increases, the proposed upper bound is difficult to calculate when $k$ is large. 
However, CC upper bound can be calculated for most $k$ because CC upper bound is parameterized for easy calculation.
In Fig. \ref{DFR_n1024}, Monte Carlo is the DFR value obtained by performing Monte Carlo simulation of NewHope protocol.
Therefore, this DFR value reflects the error dependency, but this simulation is only possible for higher noise case (i.e., larger $k$ values).
If we compare the Monte Carlo with no error dependency, it is confirmed that Monte Carlo DFR values are slightly larger than the no error dependency.
The reason for this is that NewHope uses an ECC called ATE \cite{ErrorDependency}, and therefore the DFR performance is degraded due to error dependency.
Also, according to argument in \cite{ErrorDependency}, since NewHope uses ATE as an ECC, no error dependency becomes too positive.
Fig. \ref{DFR_n1024} shows that as $k$ increases, the proposed upper bound and no error dependency become almost identical.
It is confirmed that the no error dependency is referred to as the lower bound of the DFR of a ring-LWE-based cryptosystem with an error dependency \cite{ErrorDependency}.
Therefore, it is guaranteed that the proposed upper bound is a fairly tight upper bound, especially for large $k$.

Fig. \ref{DFR_n512} compares the various upper bounds on DFR of NewHope for various noise parameter $k$ for $n=512$.
First of all, it is confirmed that the two proposed upper bounds improve the upper bound more than forty order of magnitude compared to the current upper bound for $k=8$.
Note that the proposed upper bound on DFR of NewHope is less than $10^{-120}$, the simplified upper bound is less than $10^{-111}$, and the current upper bound is less than $10^{-63}$.
Unlike the case of $n = 1024$, the proposed upper bound can be calculated for most $k$ when $n = 512$.
Thus, when $n = 512$, we can calculate tight upper bound values for most $k$.
It is confirmed that there is almost no difference between the the proposed upper bound and no error dependency, which is the lower bound of DFR of Ring-LWE based cryptosystem, for most $k$.
Therefore, it is guaranteed that the proposed upper bound is a fairly tight upper bound for most $k$.
\begin{figure}[t]
\centering
\includegraphics[width=\textwidth]{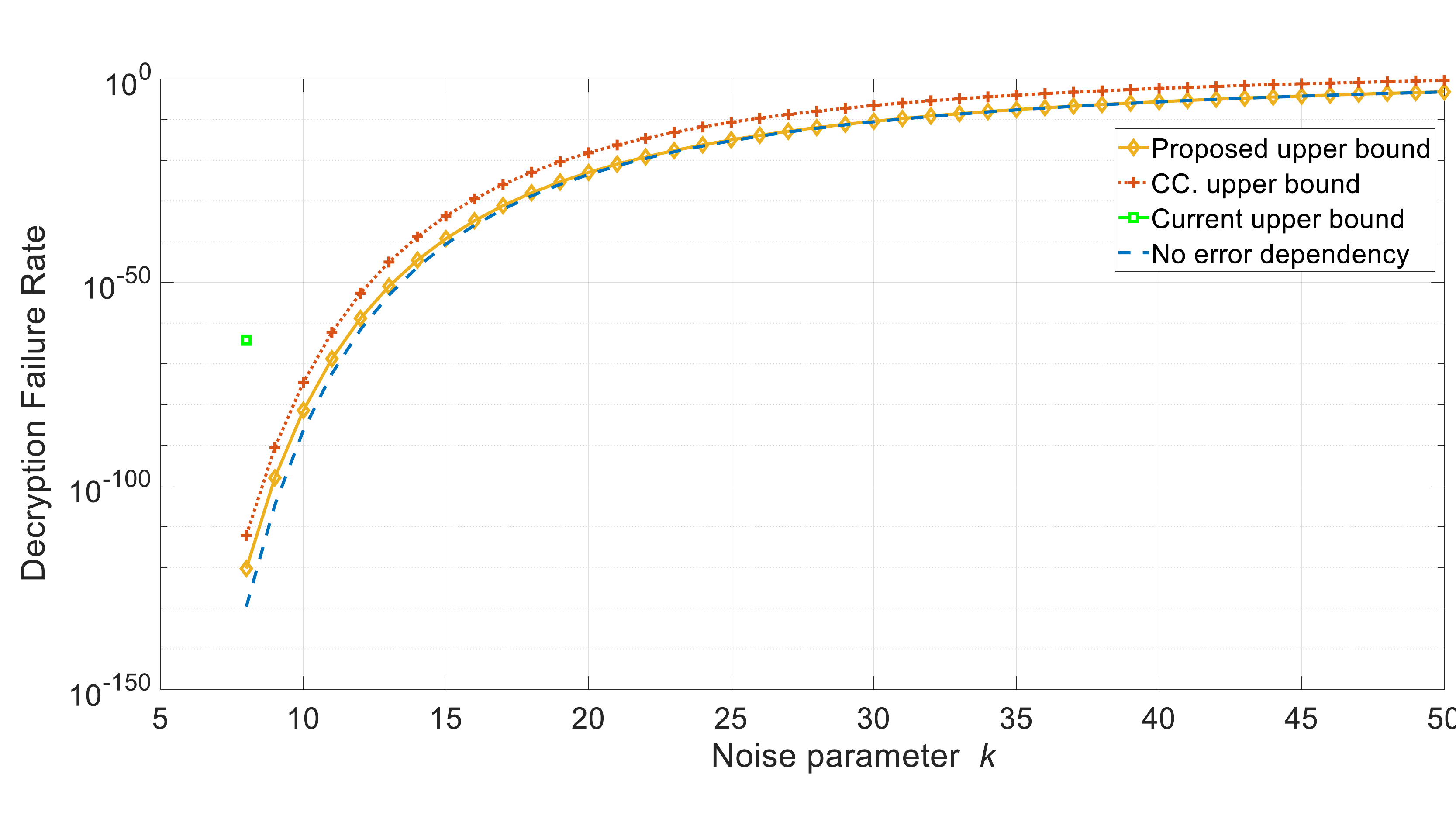}
\caption{Comparison of various upper bounds for various $k$ ($n=512$). } 
\label{DFR_n512}
\end{figure}

In conclusion, when $n = 1024$ and $n = 512$, it is confirmed that the proposed upper bound is fairly tight.
Furthermore, Figs \ref{DFR_n1024} and \ref{DFR_n512} show that when the noise parameter $k$ is $8$, the proposed upper bound on DFR of NewHope is much smaller than the DFR requirement of PQC.
Therefore, by utilizing this new DFR margin, the security and bandwidth efficiency of NewHope can be improved, which will be verified in the next section.



%

\section{Improved Security and Bandwidth Efficiency of NewHope Based on New Upper Bound on DFR}

\subsection{Improved Security}
Since there exists a trade-off relation between the security level and the DFR, it is necessary to properly select the noise parameter $k$ of centered binomial distribution such that the security level and the DFR are appropriately determined to meet the requirements.
Since it is confirmed by the new upper bound on DFR that NewHope is designed to have unnecessarily low DFR, the security level can be more improved by using the new DFR margin which is the difference between new upper bound and the required DFR.

\begin{table}[t]
\centering
\caption{Improved security level of NewHope based on new DFR margin (The noise parameter of current NewHope is $k=8$) and the required DFR is $2^{-140}$.}
\begin{tabular}{|c|c|c|c|c|}
\hline
\multicolumn{1}{|c|}{$n$}  & \multicolumn{1}{|c|}{$k$} & \multicolumn{1}{|c|}{DFR}         & \begin{tabular}[c]{@{}c@{}}Cost of primal attack\\ Classical/Quantum [bits]\end{tabular} & \begin{tabular}[c]{@{}c@{}}Cost of dual attack\\ Classical/Quantum [bits]\end{tabular} \\ \hline
\multirow{8}{*}{1024} & 8   & $ \leq 2^{-418}$ & \multicolumn{1}{|c|}{259/235}                                                            & \multicolumn{1}{|c|}{257/233}                                                         
\\ \cline{2-5}
& 9   & $ \leq 2^{-341}$ & \multicolumn{1}{|c|}{262/238}                                                            & \multicolumn{1}{|c|}{261/237}                                                         
\\ \cline{2-5}
& 10   & $ \leq 2^{-284}$ & \multicolumn{1}{|c|}{266/241}                                                            & \multicolumn{1}{|c|}{265/240}                                                         
\\ \cline{2-5}
& 11   & $ \leq 2^{-240}$ & \multicolumn{1}{|c|}{269/244}                                                            & \multicolumn{1}{|c|}{268/243}                                                         
\\ \cline{2-5}
& 12   & $ \leq 2^{-205}$ & \multicolumn{1}{|c|}{272/247}                                                            & \multicolumn{1}{|c|}{271/246}                                                         
\\ \cline{2-5}
& 13   & $ \leq 2^{-178}$ & \multicolumn{1}{|c|}{275/249}                                                            & \multicolumn{1}{|c|}{274/248}                                                         
\\ \cline{2-5}
& 14   & $ \leq 2^{-156}$ & \multicolumn{1}{|c|}{278/252}                                                            & \multicolumn{1}{|c|}{276/250}                                                         
\\ \cline{2-5}
 & 15  & $\leq 2^{-137}$ & \multicolumn{1}{|c|}{280/254}                                                            & \multicolumn{1}{|c|}{279/253}                                                          \\ \hline
\multirow{8}{*}{512}  & 8  & $\leq 2^{-399}$ & \multicolumn{1}{|c|}{112/101}                                                            & \multicolumn{1}{|c|}{112/101}                                                          \\ \cline{2-5}
& 9  & $\leq 2^{-325}$ & \multicolumn{1}{|c|}{114/103}                                                            & \multicolumn{1}{|c|}{113/103}                                                          \\ \cline{2-5}
& 10  & $\leq 2^{-270}$ & \multicolumn{1}{|c|}{115/105}                                                            & \multicolumn{1}{|c|}{115/104}                                                          \\ \cline{2-5}
& 11  & $\leq 2^{-228}$ & \multicolumn{1}{|c|}{117/106}                                                            & \multicolumn{1}{|c|}{117/106}                                                          \\ \cline{2-5}
& 12  & $\leq 2^{-195}$ & \multicolumn{1}{|c|}{119/107}                                                            & \multicolumn{1}{|c|}{118/107}                                                          \\ \cline{2-5}
& 13  & $\leq 2^{-169}$ & \multicolumn{1}{|c|}{120/109}                                                            & \multicolumn{1}{|c|}{119/108}                                                          \\ \cline{2-5}
& 14  & $\leq 2^{-147}$ & \multicolumn{1}{|c|}{121/110}                                                            & \multicolumn{1}{|c|}{121/110}                                                          \\ \cline{2-5}
  & 15  & $\leq 2^{-130}$ & \multicolumn{1}{|c|}{122/111}                                                            & \multicolumn{1}{|c|}{122/111}                                                          \\ \hline
\end{tabular}
\label{NewHope Security level}
\end{table}

Table \ref{NewHope Security level} shows the improved security levels which are calculated as the cost of the primal attack and the cost of dual attack \cite{Attack} to NewHope. 
It is possible to improve the security level by 7.2 \% ($n=1024$, $k=14$) and 8.9 \% ($n=512$, $k=14$) while guaranteeing the required DFR of $2^{-140}$ compared with the current NewHope.
Note that such security level improvement does not require much increase of time/space complexity in NewHope because it only changes the noise parameter $k$ without any additional procedure.
Therefore, this improvement of security can be easily applied to NewHope.



\subsection{Improved Bandwidth Efficiency}
The bandwidth efficiency of NewHope can also be improved by utilizing new DFR margin. 
An improvement of bandwidth efficiency is achieved by reducing (or more compressing) the ciphertext size which, however, increases the compression noise resulting in the DFR degradation.
Even with such increased compression noise, both the improvement of bandwidth efficiency and the required DFR of $2^{-140}$ can be achieved by utilizing new DFR margin.



\begin{table}[t]
\caption{Improved bandwidth efficiency of NewHope based on new DFR margin (The noise parameter and compression rate of current NewHope are $k=8$ and $r=8$, respectively and the required DFR is $2^{-140}$.).}
\centering
\begin{tabular}{|c|c|c|c|c|}
\hline
$n$                   & $r$                & $k$ & Ciphertext reduction (\%) & DFR             \\ \hline
\multirow{4}{*}{1024} & 8                  & 8   & 0 (Current NewHope)       & $\leq 2^{-418}$ \\ \cline{2-5} 
                      & \multirow{3}{*}{4} & 8   & 5.9                       & $\leq 2^{-212}$ \\ \cline{3-5} 
                      &                    & 9   & 5.9                       & $\leq 2^{-173}$ \\ \cline{3-5} 
                      &                    & 10  & 5.9                       & $\leq 2^{-144}$ \\ \hline
\multirow{3}{*}{512}  & 8                  & 8   & 0 (Current NewHope)       & $\leq 2^{-399}$ \\ \cline{2-5} 
                      & \multirow{2}{*}{4} & 8   & 5.9                       & $\leq 2^{-199}$ \\ \cline{3-5} 
                      &                    & 9   & 5.9                       & $\leq 2^{-161}$ \\ \hline
\end{tabular}
\label{NewHope Bandwidth efficiency}
\end{table}

Table \ref{NewHope Bandwidth efficiency} shows the improved bandwidth efficiency of NewHope achieved by additional ciphertext compression. 
It is possible to improve the bandwidth efficiency by 5.9 \% by changing the compression rate on $ v '$ from 8 (3 bits per coefficient) to 4 (2 bits per coefficient) and the security level by 2.5 \% by changing the noise parameter from 8 to 10 for $n=1024$.
Similarly, it is possible to improve the bandwidth efficiency by 5.9 \% and the security level by 1.9 \% by changing the noise parameter from 8 to 9 for $n=512$.
The improvement of the security and bandwidth efficiency requires little change in the protocol of NewHope, so that this improvement can be easily applied to NewHope.


\subsection{Closeness of Centered Binomial Distribution and the Corresponding Rounded Gaussian Distribution for Various $k$}

The properties of rounded Gaussian distribution $\xi$ are key factor to the worst-case to average-case reduction for Ring-LWE.
However, since a very high-precision and high-complexity sampling is required for the rounded Gaussian distribution, NewHope uses the centered binomial distribution $\psi_k$ for practical sampling without having rigorous security proof.
It is generally accepted that as the centered binomial distribution and the rounded Gaussian distribution are closer to each other, NewHope is regarded as more secure. 
The closeness of two distribution can be measured through many methods.
Among them, R{\'e}nyi divergence is a well-known method, which is parameterized by a real $a>1$ and defined for two distributions $P$ and $Q$ as follows \cite{Renyi1}, \cite{Renyi2}.

\begin{equation}
R_a(P||Q)=\Bigg( \sum_{x \in \sup(P)} \frac{P(x)^a}{Q(x)^{a-1}} \Bigg)^{\frac{1}{a-1}}
\end{equation}
where $sup(P)$ represents the support of $P$ and $Q(x) \neq 0$ for $x \in sup(P)$.

We define $\xi_k$ to be the rounded Gaussian distribution with the variance $\sigma^2=k/2$, which is the distribution of $\lfloor \sqrt{k/2} \cdot x \rceil$ where $x$ follows the standard normal distribution.

\begin{figure}[h]
\centering
\includegraphics[width=\textwidth]{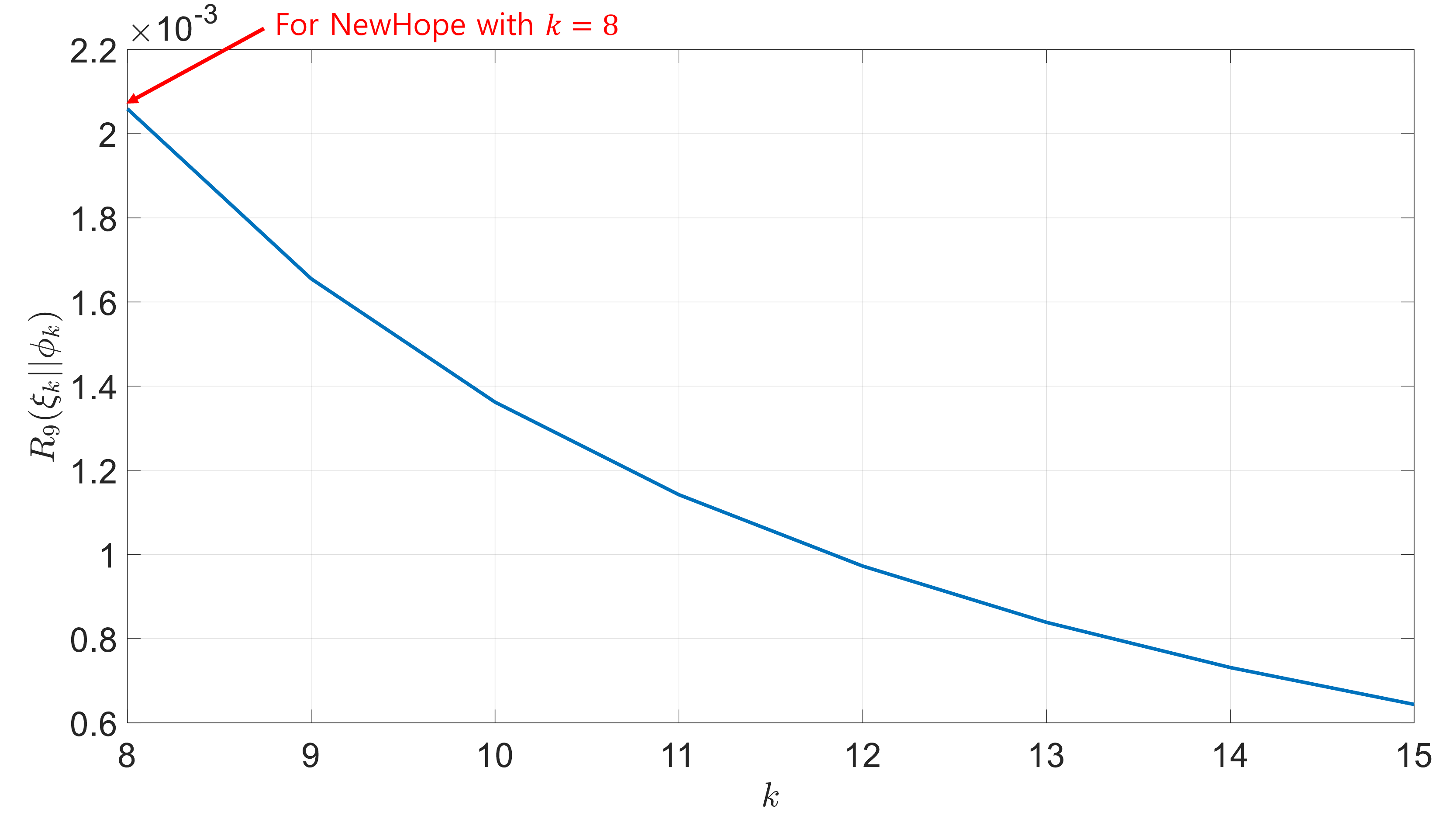}
\caption{R{\'e}nyi divergence of the centered binomial distribution $\psi_k$ and the rounded Gaussian distribution $\xi_k$ with the same variance $k/2$ according to $k$ ($a=9$).} 
\label{Renyi}
\end{figure}

Fig. \ref{Renyi} shows that the R{\'e}nyi divergence ($a=9$ is used as in \cite{NewHope NIST}) of the centered binomial distribution $\psi_k$ and the rounded Gaussian distribution $\xi_k$ with the same variance $k/2$.
It is clear that the R{\'e}nyi divergence decreases as $k$ increases.
Therefore, an increase in the noise parameter $k$ can quantitatively and qualitatively improve the security of NewHope although the time complexity increases a little bit due to the complexity increase of calculating $\sum_{i=0}^{k-1}(b_i-b_i')$.


\section{Conclusions}
Since NewHope is an IND-CCA secure KEM by applying the FO transform to an IND-CPA secure PKE, accurate DFR calculation is required to guarantee resilience against attacks that exploit decryption failures.
However, the upper bound on DFR of NewHope derived in \cite{NewHope NIST}, \cite{NewHope1} is rather loose because the compression noise and effect of encoding/decoding of ATE in NewHope are not fully considered.
Also, the centered binomial distribution is approximated by subgaussian distribution.
Furthermore, since NewHope is a Ring-LWE based cryptosystem, there is a problem of error dependency among error coefficients, which makes accurate DFR calculation difficult.

In this paper, an upper bound on DFR, which is much closer to the real DFR than previous upper bound on DFR derived in \cite {NewHope NIST} , \cite {NewHope1}, is derived by considering the above-ignored factors.
Also, the centered binomial distribution is not approximated by the subgaussian distribution.
Especially, the new upper bound on DFR considers the error dependency among error coefficients by using the constraint relaxation and union bound.
Furthermore, the new upper bound on DFR is parameterized by using CC bound in order to facilitate calculation of new upper bound on DFR for the parameters of NewHope.

According to the new upper bound on DFR of NewHope, since it is much lower than the DFR requirement of PQC, this DFR margin can be used to improve the security and bandwidth efficiency.
As a result, the security level of NewHope is improved by 7.2\%, or the bandwidth efficiency is improved by 5.9\%.
This improvement in the security and bandwidth efficiency can be easily achieved in NewHope because there is little change in time/space complexity of NewHope.

\end{document}